\documentclass[orivec]{llncs}
\usepackage{amssymb}
\usepackage{amsmath}
\usepackage{qtree}
\usepackage{color}
\usepackage{graphicx}
\usepackage{tikz}
\usetikzlibrary{arrows}

\newcommand{\args}[1]{\vec{#1}}
\newcommand{\bwedge}[1]{\bigwedge_{#1}}
\newcommand{\lvar}{\mathrm{lvar}}

\newcommand{\Ind}{\mathcal{I}}
\newcommand{\V}{\mathcal{V}}

\newcommand{\T}{\mathcal{T}}
\newcommand{\MF}{\mathcal{F}}
\newcommand{\MR}{\mathcal{P}}

\newcommand{\norm}[1]{ {#1}\!\!\downarrow}
\newcommand{\sat}{\vDash}
\newcommand{\unsat}{\not \vDash}
\newcommand{\imp}{\rightarrow}
\newcommand{\R}[1]{\mu^{-1}_{#1}}
\newcommand{\Proj}[1]{\mu_{#1}}
\newcommand{\Refine}{ \operatorname{Ref}}
\newcommand{\refine}{ \Rightarrow_{\operatorname{Ref}}}

\newcommand{\mo}{ \Rightarrow_{\operatorname{MO}}}
\newcommand{\li}{ \Rightarrow_{\operatorname{LI}}}
\newcommand{\sh}{ \Rightarrow_{\operatorname{SH}}}
\newcommand{\des}[3]{ \Rightarrow^{#3}_{#1#2}}

\newcommand{\apr}{ \Rightarrow_{\operatorname{AP}}}

\newcommand{\anc}{ \Rightarrow_{\operatorname{A}}}

\newcommand{\I}[1]{\mathcal{I}_{#1}}

\newcommand{\pos}{\mathrm{pos}}
\newcommand{\sel}{\mathrm{sel}}
\newcommand{\mgu}{\mathrm{mgu}}  
\newcommand{\skt}{\mathrm{skt}}
 
\newcommand{\ground}[1]{\mathcal{G}(#1)} 
\newcommand{\vars}{\operatorname{vars}}

\newcommand{\shortrules}[6]{\noindent\begin{minipage}{#6ex}{\bfseries #1}\end{minipage} $\;$ #2 $\;\Rightarrow_{\operatorname{#5}}\;$ #3 \par\smallskip\noindent #4} 

\begin{document}

\title{Decidability of the Monadic Shallow Linear First-Order Fragment with Straight Dismatching Constraints}

\author{Andreas Teucke\inst{1,2} \and Christoph Weidenbach\inst{1}}

\institute{Max-Planck Institut f\"ur Informatik, Saarland Informatics Campus,
	66123 Saarbr\"ucken
	Germany \and Graduate School of Computer Science, Saarbr\"ucken, Germany
}

\maketitle

\begin{abstract}
The monadic shallow linear Horn fragment is well-known to be decidable and
has many application, e.g., in security protocol analysis, tree automata,
or abstraction refinement. It was a long standing open problem how to extend
the fragment to the non-Horn case, preserving decidability, that would, e.g., enable to express non-determinism
in protocols. We prove decidability of the non-Horn monadic shallow linear
fragment via ordered resolution further extended with dismatching constraints
and discuss some applications of the new decidable fragment.
\end{abstract}

\section{Introduction} \label{sec:intro}

Motivated by the automatic analysis of security protocols,
the monadic shallow linear Horn (MSLH) fragment was shown to be decidable in \cite{Weidenbach99cade}.
In addition to the restriction to monadic Horn clauses, the main restriction of the fragment
is positive literals of the form $S(f(x_1,\ldots,x_n))$ or $S(x)$ where all $x_i$ are different, i.e., all terms are shallow
and linear. The fragment can be finitely saturated by superposition (ordered resolution) where negative literals
with non-variable arguments are always selected. As a result, productive clauses with respect to the 
superposition model operator $\I{N}$ have the form $S_1(x_1),\ldots,S_n(x_n) \rightarrow S(f(x_1,\ldots,x_n))$.
Therefore, the models of saturated MSLH clause sets can both be represented by tree automata \cite{tata2007} and shallow linear
sort theories \cite{JacquemardMeyerEtAl98}. The models are typically infinite. The decidability result of MSLH clauses was rediscovered in the
context of tree automata research \cite{GoubaultLarrecq05IPL} where in addition DEXPTIME-completeness of the MSLH
fragment was shown. The fragment was further extended by disequality constraints \cite{journalsiplSeidlR11,Seidl12F} still motivated
by security protocol analysis \cite{SeidlVerma2007}. Although from a complexity point of view, the difference between
Horn clause fragments and the respective non-Horn clause fragments is typically reflected by membership in the deterministic vs.\ the
non-deterministic respective complexity fragment, for monadic shallow linear clauses so far there was no decidability 
result for the non-Horn case.

The results of this paper close this gap. We show the monadic shallow linear non-Horn (MSL) clause fragment
to be decidable by superposition (ordered resolution). From a security protocol application point of view,
non-Horn clauses enable a natural representation of non-determinism. Our second extension to the fragment are
unit clauses with disequations of the form $s\not\approx t$, where $s$ and $t$ are not unifiable. Due to the
employed superposition calculus, such disequations do not influence saturation of an MSL clause set, but have an effect
on potential models. They can rule out identification of syntactically different ground terms as it is, e.g., desired in the security
protocol context for syntactically different messages or nonces. Our third extension to the fragment 
are straight dismatching constraints. These constraints are incomparable to the disequality constraints mentioned above \cite{journalsiplSeidlR11,Seidl12F}.
They do not strictly increase the expressiveness of the MSL theory, but enable up to exponentially more compact
saturations. For example, the constrained clause\newline
\centerline{$(S(x), T(y) \rightarrow S(f(x,y)); y\neq f(x',f(a,y')))$}
over constants $a, b$ describes the same set of ground clauses as the six unconstrained clauses\newline
\centerline{$S(x), T(a) \rightarrow S(f(x,a)) \qquad S(x), T(b) \rightarrow S(f(x,b)) \qquad \ldots$}
\centerline{$S(x), T(f(b,y')) \rightarrow S(f(x,f(b,y')))$} 
\centerline{$S(x), T(f(f(x'',y''),y')) \rightarrow S(f(x,f(f(x'',y''),y'))$.} 
Furthermore, for a  satisfiability equivalent transformation into MSL clauses,
the nested terms in the positive literals would have to be factored out by the introduction of further predicates and clauses.
E.g., the first clause is replaced by the two MSL clauses $S(x), T(a), R(y) \rightarrow S(f(x,y))$ and $R(a)$ where
$R$ is a fresh monadic predicate. The constrained clause belongs to the MSL(SDC) fragment.
Altogether, the resulting MSL(SDC) fragment is shown to be decidable in Section~\ref{sec:decide}.

The introduction of straight dismatching constraints (SDCs) enables an improved refinement step of our
approximation refinement calculus~\cite{Teucke2015}. Before, several clauses were needed to rule
out a specific instance of a clause in an unsatisfiable core.
For example, if due to a linearity approximation from clause $S(x), T(x) \rightarrow S(f(x,x))$ to
$S(x), T(x), S(y), T(y) \rightarrow S(f(x,y))$ an instance $\{x\mapsto f(a,x')$, $y\mapsto f(b,y')\}$ is
used in the proof, before \cite{Teucke2015} several clauses were needed to replace $S(x), T(x) \rightarrow S(f(x,x))$ in
a refinement step in order to rule out this instance. With straight dismatching constraints the clause
$S(x), T(x) \rightarrow S(f(x,x))$ is replaced by the  two clauses
$S(f(a,x)), T(f(a,x)) \rightarrow S(f(f(a,x),f(a,x)))$ and $(S(x), T(x) \rightarrow S(f(x,x)); x\neq f(a,y))$.
For the improved approximation refinement approach (FO-AR) presented in this paper,
any refinement step results in just two clauses, see Section~\ref{sec:approx}. The additional expressiveness of constraint clauses comes
almost for free, because necessary computations, like, e.g., checking emptiness of SDCs, can all be done in polynomial time,
see Section~\ref{sec:prelim}.

In addition to the extension of the known MSLH decidability result and the improved approximation refinement calculus FO-AR,
we discuss in Section~\ref{sec:experiments} the potential of the MSL(SDC) fragment in the context of
FO-AR, Theorem~\ref{theo:refinement:scfoar}, and its prototypical implementation in SPASS-AR (\url{http://www.mpi-inf.mpg.de/fileadmin/inf/rg1/spass-ar.tgz}).
It turns out that for clause sets containing certain structures, FO-AR
is superior to ordered resolution/superposition~\cite{BachmairGanzinger94b} and instance generating methods~\cite{Korovin13ganzinger}. 
The paper ends with a discussion
on challenges and future research directions, Section~\ref{sec:conclusion}.

\section{First-Order Clauses with Straight Dismatching Constraints:  MSL(SDC)} \label{sec:prelim}

We consider a standard first-order language where
letters $v,w,x,$ $y,z$ denote variables, $f,g,h$ functions, $a,b,c$ constants, $s,t$ terms, $p,q,r$ positions and Greek letters $\sigma,\tau,\rho,\delta$ are used for substitutions. 
$S,P,Q,R$ denote predicates, $\approx$ denotes equality, $A,B$ atoms, $E,L$ literals, $C,D$ clauses, $N$ clause sets  and $\V$ sets of variables.
$\overline L$ is the complement of $L$.
The signature $\Sigma=(\MF,\MR)$ consists of two disjoint, non-empty, in general infinite sets of  function and predicate symbols
$\MF$ and $\MR$, respectively.
The set of all \emph{terms} over variables $\V$ is $\T(\MF,\V)$.
If there are no variables, then terms, literals and clauses are called \emph{ground}, respectively.
A \emph{substitution} $\sigma$ is denoted by pairs $\{x \mapsto t\}$ and its update at $x$  by  $\sigma[ x \mapsto t]$.
A substitution $\sigma$ is a \emph{grounding} substitution for $\V$ if $x\sigma$ is ground for every variable $x \in \V$.

The set of \emph{free} variables of an atom $A$ (term $t$) denoted by $\vars(A)$ ($\vars(t)$).
A \emph{position} is a sequence of positive integers, where $\varepsilon$ denotes the empty position.
As usual $t\vert_p = s$ denotes the subterm $s$ of $t$ at position $p$, which we also write as $t[s]_p$, 
and $t[p/s']$ then denotes the replacement of $s$ with $s'$ in $t$ at position $p$.
These notions are extended to literals and multiple positions.

A predicate with exactly one argument is called \emph{monadic}. 
A term is \emph{complex} if it is not a variable and \emph{shallow} if it has at most depth one. 
It is called \emph{linear} if there are no duplicate variable occurrences. 
A literal, where every argument term is shallow, is also called \emph{shallow}. 
A variable and a constant are called \emph{straight}.
A term $f(s_1,\ldots,s_n)$ is called \emph{straight}, if $s_1,\ldots,s_n$ are different variables except for at most one straight term $s_i$. 

A \emph{clause} is a multiset of literals which we write as an implication $\Gamma \imp \Delta$ 
where the atoms in the multiset $\Delta$ (the\emph{ succedent}) denote the positive literals and the atoms in the multiset $\Gamma$ (the \emph{antecedent}) the negative literals.
We write $\square$ for the empty clause.
If $\Gamma$ is empty we omit $\imp$, e.g., we can write $P(x)$ as an alternative of $\imp P(x)$.
We abbreviate disjoint set union with sequencing, for example,  we write $\Gamma,\Gamma' \imp \Delta,L$ instead 
of $\Gamma \cup \Gamma' \imp \Delta \cup \{L\}$.
A clause $E,E,\Gamma \imp \Delta$ is equivalent to $E,\Gamma \imp \Delta$ and we call them equal \emph{modulo duplicate literal elimination}. 
If every term in $\Delta$ is shallow, the clause is called \emph{positive shallow}.
If all atoms in $\Delta$ are linear and variable disjoint, the clause is called \emph{positive linear}.
A clause $\Gamma \imp \Delta$ is called an \emph{MSL} clause, if it is (i)~positive shallow and linear, (ii)~all occurring predicates are monadic,
(iii)~no equations occur in $\Delta$, and (iv)~no equations occur in $\Gamma$ or $\Gamma = \{s\approx t\}$ and $\Delta$ is empty
where $s$ and $t$ are not unifiable.
\emph{MSL} is the first-order clause fragment consisting of MSL clauses. 
Clauses $\Gamma,s\approx t \imp \Delta$ where $\Gamma$, $\Delta$ are non-empty and $s,t$ are not unifiable 
could be added to the MSL fragment without changing any of our results.
Considering the superposition calculus, it will select $s\approx t$.
Since the two terms are not unifiable, no inference will take place on such a clause and the clause will not
contribute to the model operator. 
In this sense such clauses do not increase the expressiveness of the fragment.

An \emph{atom ordering} $\prec$ is an irreflexive, well-founded, total ordering on ground atoms.  
It is lifted to literals by representing $A$ and $\neg A$ as multisets $\{A\}$ and $\{A,A\}$, respectively.
The multiset extension of the literal ordering induces an ordering on ground clauses.
The clause ordering is compatible with the atom ordering; if the maximal atom in $C$ is greater than the maximal atom in $D$ then $D \prec C$.
We use  $\prec$ simultaneously to denote an atom ordering and its multiset, literal, and clause extensions.
For a ground clause set $N$ and clause $C$, the set $N^{\prec C}=\{D \in N \mid D \prec C\}$ denotes the clauses of $N$ smaller than $C$.

A \emph{Herbrand interpretation} $\I{}$ is a - possibly infinite - set of ground atoms.
A ground atom $A$ is called \emph{true} in $\I{}$ if $A\in\I{}$ and \emph{false}, otherwise.
$\I{}$ is said to \emph{satisfy} a ground clause $C= \Gamma \imp \Delta$, denoted by $\I{}\sat C$, if   $\Delta \cap \I{} \neq \emptyset$ or $\Gamma \not\subseteq \I{}$.
A non-ground clause $C$ is satisfied by $\I{}$ if $\I{}\sat C\sigma$ for every grounding substitution $\sigma$.
An interpretation $\I{}$ is called a \emph{model} of $N$, $\I{}\sat N$, if $\I{}\sat C$ for every $C\in N$. 
A model $\I{}$ of $N$ is considered \emph{minimal} with respect to set inclusion, i.e., if there is no model $\I{}'$ with $\I{}'\subset \I{}$ and $\I{}'\sat N$.
A set of clauses $N$ is \emph{satisfiable}, if there exists a model that satisfies $N$. 
Otherwise, the set is \emph{unsatisfiable}.

A disequation $t \neq s$ is an \emph{atomic straight dismatching constraint}
if $s$ and $t$ are variable disjoint terms and $s$ is straight.
A straight dismatching constraint $\pi$ is a conjunction of atomic straight dismatching constraints.
Given a substitution $\sigma$, $\pi\sigma= \bwedge{i\in I} ~~ t_i\sigma \neq s_i $.
$\lvar(\pi) := \bigcup_{i \in I} \vars(t_i)$ are the left-hand variables of $\pi$
and the depth of $\pi$  is the maximal term depth of the $s_i$.
A \emph{solution} of $\pi$ is a grounding substitution $\delta$  
such that for all $i\in I$, $t_i\delta$ is not an instance of $s_i$,
i.e., there exists no $\sigma$ such that $t_i\delta = s_i\sigma$.
A dismatching constraint is solvable if it has a solution and unsolvable, otherwise.
Whether a straight dismatching constraint is solvable, is decidable in linear-logarithmic time \cite{DBLP:conf/cade/TeuckeW16}.
$\top$ and $\bot$ represent the true and false dismatching constraint, respectively. 

 We define constraint normalization $\norm{\pi}$ as the normal form of the following rewriting rules over straight dismatching constraints.
 
 \shortrules{\ }{$ ~\pi \wedge f(t_1,\ldots,t_n)\neq y  $~~~~~~~~~~~~~~~}{$\bot$}{}{}{10}
 \shortrules{\ }{$ \pi \wedge f(t_1,\ldots,t_n)\neq f(y_1,\ldots,y_n)  $}{$\bot$}{}{}{10}
 \shortrules{\ }{$ \pi \wedge f(t_1,\ldots,t_n)\neq f(s_1,\ldots,s_n) $}{$\pi \wedge t_i \neq s_i$  ~~if $s_i$ is complex}{}{}{10}
 \shortrules{\ }{$ \pi \wedge f(t_1,\ldots,t_n)\neq g(s_1,\ldots,s_m)  $}{$\pi$}{}{}{10}
 \shortrules{\ }{$\pi \wedge  x\neq s \wedge x\neq s\sigma $~~~~~~~~~~~~~~~~~\;}{$\pi \wedge  x\neq s $  }{}{}{10}
Note that $f(t_1,\ldots,t_n)\neq f(s_1,\ldots,s_n)$ normalizes to $t_i \neq s_i$ for some $i$,
 where $s_i$ is the one straight complex argument of $f(s_1,\ldots,s_n)$.
Furthermore, the depth of $\norm{\pi}$ is less or equal to the depth of $\pi$ and both have the same solutions.

A pair of a clause and a constraint  $(C;\pi)$ is called a \emph{constrained clause}.
Given a substitution $\sigma$, $(C;\pi)\sigma = (C\sigma;\pi\sigma)$.
$C\delta$ is called a ground clause of $(C;\pi)$ if $\delta$ is a solution of $\pi$. 
$\ground{(C;\pi)}$  is the set of ground instances of $(C;\pi)$. 
If $\ground{(C;\pi)}\subseteq \ground{(C';\pi')}$, then $(C;\pi)$ is an instance of $(C';\pi')$.
If $\ground{(C;\pi)} = \ground{(C';\pi')}$, then $(C;\pi)$ and $(C';\pi')$ are called variants.
A Herbrand interpretation $\I{}$ satisfies $(C;\pi)$, if  $\I{}\sat \ground{(C;\pi)}$.
A constrained clause $(C;\pi)$ is called \emph{redundant} in $N$ 
if for every $D \in \ground{(C;\pi)}$, 
there exist $D_1,\ldots,D_n$ in $\ground{N}^{\prec D}$ such that $D_1,\ldots,D_n \sat D$.
A constrained clause $(C';\pi')$ is called a \emph{condensation} of $(C;\pi)$ if   $C' \subset C $ and
there exists a substitution $\sigma$ such that, $\pi\sigma= \pi'$, $\pi' \subseteq \pi$, and
for all $L \in C$ there is an $L'\in C'$ with $L\sigma=L'$.
A finite unsatisfiable subset of $\ground{N}$ is called an unsatisfiable core of $N$.
 	
An MSL clause with straight dismatching constraints is called an \emph{MSL(SDC)} clause 
with MSL(SDC) being the respective first-order fragment. 
Note that any clause set $N$ can be transformed into an equivalent constrained clause set 
by changing each $C\in N$ to $(C;\top)$.

\section{Decidability of the MSL(SDC) fragment}\label{sec:decide}

In the following we will show that the satisfiability of the MSL(SDC) fragment is decidable.
For this purpose we will define ordered resolution with selection on constrained clauses \cite{DBLP:conf/cade/TeuckeW16} and
show that with an appropriate ordering and selection function, saturation of an MSL(SDC) clause set terminates.

For the rest of this section we assume an atom ordering $\prec$ such that a literal $\neg Q(s)$ is not greater than a literal $P(t[s]_p)$, 
where $p \neq \varepsilon$. For example, a KBO where all symbols have weight one
has this property.

\begin{definition}[sel]\label{def:decide:sel}
Given an MSL(SDC) clause  $(C;\pi) = (S_1(t_1),\dots,S_n(t_n) \imp P_1(s_1),\dots, P_m(s_m);\pi)$. 
The Superposition Selection function $\sel$ is defined by $S_i(t_i)\in \mathrm{sel}(C)$ if
(1) $t_i$ is not a variable or
(2) $t_1,\ldots,t_n$ are variables and $t_i \notin \vars(s_1,\dots,s_m)$ or
(3) $\{ t_1,\ldots,t_n\} \subseteq \vars(s_1,\dots,s_m)$ and for some $1 \leq j \leq m$, $s_j=t_i$.
\end{definition}

The selection function $\sel$ (Definition \ref{def:decide:sel}) ensures that a clause $\Gamma \imp \Delta$ can only be resolved on a positive literal 
if $\Gamma$ contains only variables, which also appear in $\Delta$ at a non-top position.
For example:\newline
\centerline{$\begin{array}{r@{\,=\,}l}
\sel(P(f(x)),P(x), Q(z) \imp Q(x),R(f(y)) & \{P(f(x))\}\\
\sel(P(x), Q(z) \imp Q(x),R(f(y))) & \{Q(z)\}\\
\sel(P(x), Q(y) \imp Q(x),R(f(y))) & \{P(x)\}\\
\sel(P(x), Q(y) \imp Q(f(x)),R(f(y))) & \emptyset.\\
\end{array}$}
Note that given an MSL(SDC) clause $(C;\pi) = (S_1(t_1),\dots,S_n(t_n)$ $ \imp P_1(s_1),\dots P_m(s_m);\pi)$,
if some $S_i(t_i)$ is maximal in $C$, then at least one literal is selected. 

\begin{definition}
A literal $A$ is called \emph{[strictly] maximal} in a constrained clause $(C \vee A;\pi)$ 
if and only if there exists a solution $\delta$ of $\pi$ such that for all literals $B$ in $C$,
$B\delta \preceq A\delta~ [ B\delta \prec A\delta ]$.
\end{definition}

\begin{definition}[SDC-Resolution]\label{def:decide:resolution}
$$ { \frac{( \Gamma_1 \imp \Delta_1, A ~;~\pi_1) \qquad ( \Gamma_2 , B \imp \Delta_2 ~;~\pi_2)}{ ((\Gamma_1,\Gamma_2 \imp \Delta_1,\Delta_2)\sigma~; ~\norm{(\pi_1 \wedge \pi_2)\sigma}) }}~~~, \text{if} $$
\begin{tabular}{rlrl}
  1. & $\sigma=\mgu(A,B)$ & 2. & $\norm{(\pi_1 \wedge \pi_2)\sigma}$ is solvable\\
  3. & \multicolumn{3}{l}{$A\sigma$ is strictly maximal in $(\Gamma_1 \imp \Delta_1, A;\pi_1)\sigma$ and $\sel(\Gamma_1 \imp \Delta_1, A)=\emptyset$}\\
  4. & $B \in \sel(\Gamma_2 , B \imp \Delta_2)$ & & \\
  5. & \multicolumn{3}{l}{$\sel(\Gamma_2 , B \imp \Delta_2)=\emptyset$ and $\neg B\sigma$ maximal in $(\Gamma_2 , B \imp \Delta_2;\pi_2)\sigma$}\\
\end{tabular}
\end{definition}

\begin{definition}[SDC-Factoring]\label{def:decide:factoring}
$$ {\frac{(\Gamma\imp \Delta, A, B ~;~ \pi) }{ ((\Gamma\imp \Delta, A)\sigma; \norm{\pi\sigma})}} ~~~, \text{if}$$
\begin{tabular}{rl@{$\quad$}rl}
  1. &  $\sigma=\mgu(A,B)$ &
 2. &  $\sel(\Gamma\imp \Delta, A, B)=\emptyset$\\
3. &  $A\sigma$ is maximal in $(\Gamma\imp \Delta, A, B;\pi)\sigma$ &
4. &  $\norm{\pi\sigma}$  is solvable\\
\end{tabular}
\end{definition}

Note that while the above rules do not operate on equations,
we can actually allow unit clauses that consist of non-unifiable disequations, i.e.,
clauses $s \approx t \imp$ where $s$ and $t$ are not unifiable.
There are no potential superposition inferences on such clauses as long as there
are no positive equations. So resolution and factoring suffice for completeness.
Nevertheless, clauses such as $s \approx t \imp$ affect the models of satisfiable problems. 
Constrained Resolution and Factoring are sound.

\begin{lemma}[Soundness]
	SDC-Resolution and SDC-Factoring are sound. 
\end{lemma}
\begin{proof}
	Let $(\Gamma_1,\Gamma_2 \imp \Delta_1,\Delta_2)\sigma\delta$ be a ground instance of $((\Gamma_1,\Gamma_2 \imp \Delta_1,\Delta_2)\sigma; (\pi_1 \wedge \pi_2)\sigma)$. 
	Then, $\delta$ is a solution of $(\pi_1 \wedge \pi_2)\sigma$ and $\sigma\delta$ is a solution of $\pi_1$ and $\pi_2$. 
	Hence, $(\Gamma_1 \imp \Delta_1, A)\sigma\delta$ and $(\Gamma_2 , B \imp \Delta_2)\sigma\delta$ are ground instances of $(\Gamma_1 \imp \Delta_1, A;\pi_1)$ and $(\Gamma_2 , B \imp \Delta_2;\pi_2)$, respectively.
	Because $A\sigma\delta= B \sigma\delta $, if $(\Gamma_1 \imp \Delta_1, A)\sigma\delta$ and $(\Gamma_2 , B \imp \Delta_2)\sigma\delta$ are satisfied, then $(\Gamma_1,\Gamma_2 \imp \Delta_1,\Delta_2)\sigma\delta$ is also satisfied.
	Therefore, SDC-Resolution is sound. 
	Let $(\Gamma\imp \Delta, A)\sigma\delta$ be a ground instance of $((\Gamma\imp \Delta, A)\sigma; \pi\sigma)$.  
	Then, $\delta$ is a solution of $\pi\sigma$ and $\sigma\delta$ is a solution of $\pi$. 
	Hence, $(\Gamma\imp \Delta, A, B )\sigma\delta$ is a ground instance of $(\Gamma\imp \Delta, A, B ;\pi)$.
	Because $A\sigma\delta= B \sigma\delta $, if $(\Gamma\imp \Delta, A, B )\sigma\delta$ is satisfied, then $(\Gamma\imp \Delta, A)\sigma\delta$ is also satisfied.
	Therefore, SDC-Factoring is sound.\qed
\end{proof}

\begin{definition}[Saturation]\label{def:decide:saturated}
A constrained clause set $N$ is called saturated up to redundancy, if for every inference between clauses in $N$
the result $(R;\pi)$ is either redundant in $N$ or $\ground{(R;\pi)} \subseteq \ground{N}$.
\end{definition}

Note that our redundancy notion includes condensation and  the condition $\ground{(R;\pi)} \subseteq \ground{N}$ allows ignoring variants
of clauses.

\begin{lemma}\label{lem:decide:condensation}
	Let constrained clause $(C';\pi')$ be a condensation of constrained clause $(C;\pi)$.
	Then, (i)$(C;\pi) \sat (C';\pi')$ and (ii)$(C;\pi)$ is redundant in $\{(C';\pi')\}$.
\end{lemma}

\begin{proof}
	Let $\sigma$ be a substitution such that $C' \subset C$, $\pi\sigma= \pi'$,
	$\pi' \subseteq \pi$, and for all $L \in C$ there is a $L'\in C'$ with $L\sigma=L'$.
	
	(i) Let $C'\delta \in \ground{(C';\pi')}$.
	Then $\sigma\delta$ is a solution of $\pi$ and hence $C\sigma\delta \in \ground{(C;\pi)}$.
	Let $\I{} \sat C\sigma\delta$. Hence, there is a $L\sigma\delta \in \I{}$ for some $L \in C$ and thus $L'\delta \in \I{}$ for some $L' \in C'$ with $L\sigma=L'$.
	Therefore, $\I{} \sat C'\delta$. Since $\I{}$ and $C'\delta$ were arbitrary, $(C;\pi) \sat (C';\pi')$.
	
	(ii) Let $C\delta \in \ground{(C;\pi)}$. Because $\pi' \subseteq \pi$, $\delta$ is a solution of $\pi'$ and hence, $C'\delta \in \ground{(C';\pi')} $.
	Therefore, since $C'\delta \subset C\delta$, $C'\delta \in \ground{\{(C';\pi')\}}^{\prec C\delta}$ and  $C'\delta\sat C\delta$.\qed
\end{proof}

\begin{definition}[Partial Minimal Model Construction]\label{def:decide:model}
Given a constrained clause set $N$, an ordering $\prec$ and the selection function $\mathrm{sel}$, we construct an interpretation $\I{N}$ for $N$, called a partial model,  inductively as follows:
\begin{align*}
 \I{C} &:= \bigcup_{D\prec C}^{D \in \ground{N}} \delta_D, \text{ where }C \in \ground{N}\\
 \delta_D &:= \left\{
  \begin{array}{l l}
    \{A\} & \quad \text{if $D=\Gamma \imp \Delta, A$  }\\
	  & \quad \text{$A$ strictly maximal, $\sel(D)=\emptyset$  and $\I{D}\unsat D$}\\
    \emptyset & \quad \text{otherwise}
  \end{array} \right. \\
\I{N} &:= \bigcup_{C \in \ground{N}} \delta_C
\end{align*}
Clauses $D$ with $\delta_D \neq \emptyset$ are called productive.
\end{definition}

\begin{lemma}[Ordered SDC Resolution Completeness]\label{lem:decide:complete}
 Let $N$ be a constrained clause set saturated up to redundancy by ordered SDC-resolution with selection. 
Then $N$ is unsatisfiable, if and only if $\square\in \ground{N}$. If $\square\not\in \ground{N}$ then $\I{N}\models N$.
\end{lemma}

\begin{proof}
	Assume $N$ is unsatisfiable but $\square\not\in \ground{N}$.
	For the partial model $\I{N}$, there exists a minimal false clause $C\sigma\in \ground{(C;\pi)}$ for some $(C;\pi)\in N$.
	
	$C\sigma$ is not productive, because otherwise $\I{N} \sat C\sigma$.
	Hence, either $\sel(C)\neq \emptyset$  or no positive literal in $C\sigma$ is strictly maximal.
	Assume  $C= \Gamma_2, B \imp \Delta_2$ with $B \in \sel(C)$ or $\neg B\sigma$ maximal. 
	Then, $B\sigma \in \I{C\sigma}$ and there exists a ground instance $(\Gamma_1 \imp \Delta_1 , A)\tau= D\tau \prec C\sigma$ of some clause $(D;\pi')\in N$,
	which produces $A\tau=B\sigma$. 
	Therefore, there exists a $\rho=\mgu(A,B)$ and ground substitution $\delta$ 
	such that $C\sigma=C\rho\delta$, $D\tau=D\rho\delta$. 
	Since $\rho\delta=\sigma$ is a solution of $\pi$ and $\pi'$, $\delta$ is a solution of $(\pi \wedge \pi')\rho$.
	Under these conditions, SDC-Resolution can be applied to $(\Gamma_1 \imp \Delta_1 , A;\pi')$ and  $(\Gamma_2, B \imp \Delta_2;\pi)$.
	Their resolvent $(R;\pi_R)=((\Gamma_1, \Gamma_2 \imp \Delta_1,  \Delta_2 )\rho;(\pi \wedge \pi')\rho)$ is either redundant in $N$ or $\ground{(R;\pi_R)} \subseteq \ground{N}$.
	Its ground instance $R\delta$ is false in $\I{N}$ and $R\delta \prec C\sigma$.
	If $(R;\pi_R)$ is redundant in $N$, there exist $C_1,\ldots,C_n$ in $\ground{N}^{\prec R\delta}$ with $C_1,\ldots,C_n \sat R\delta$.
	Because $C_i \prec R\delta \prec C\sigma$, $\I{N} \sat C_i$ and hence $\I{N} \sat R\delta$, which contradicts $\I{N} \unsat R\delta$.
	Otherwise, if  $\ground{(R;\pi_R)} \subseteq \ground{N}$, then $R\delta\in \ground{N}$, which contradicts $C\sigma$ being minimal false. 
	
	Now, assume $\sel(C)= \emptyset$ and $C= \Gamma \imp \Delta, B$ with  $B\sigma$  maximal. 
	Then, $C= \Gamma \imp \Delta', A, B$ with $A\sigma=B\sigma$. 
	Therefore, there exists a $\rho=\mgu(A,B)$ and ground substitution $\delta$ 
	such that $C\sigma=C\rho\delta$ and $\rho\delta$ is a solution of $\pi$.
	Hence, $\delta$ is a solution of $\pi\rho$.
	Under these conditions, SDC-Factoring can be applied to $(\Gamma \imp \Delta', A, B;\pi)$.
	The result $(R;\pi_R)=((\Gamma \imp \Delta', A)\rho;\pi\rho)$ 
	is either redundant in $N$ or $\ground{(R;\pi_R)} \subseteq \ground{N}$.
	Its ground instance $R\delta$ is false in $\I{N}$ and $R\delta \prec C\sigma$.
	If $(R;\pi_R)$ is redundant in $N$, there exist $C_1,\ldots,C_n$ in $\ground{N}^{\prec R\delta}$ with $C_1,\ldots,C_n \sat R\delta$.
	Because $C_i \prec R\delta \prec C\sigma$, $\I{N} \sat C_i$ and hence $\I{N} \sat R\delta$, which contradicts $\I{N} \unsat R\delta$.
	Otherwise, if $\ground{(R;\pi_R)} \subseteq \ground{N}$, then $R\delta\in \ground{N}$, which contradicts $C\sigma$ being minimal false. 
	
	Therefore, if $\square\not\in \ground{N}$, no    minimal false clause exists and    $\I{N}\models N$.            \qed                                                                           
\end{proof}

\begin{lemma}\label{lem:decide:finite_base}
Let $N$ be a set of MSL(SDC) clauses without variants or uncondensed clauses over a finite signature $\Sigma$.
$N$ is finite if there exists an integer $d$ such that for every $(C;\pi) \in N$, depth($\pi$)$\leq d$ and\\
(1) $C=  S_1(x_1),\dots,S_n(x_n),S'_1(t),\dots,S'_m(t) \imp \Delta$ or \\
(2) $C= S_1(x_1),\dots,S_n(x_n),S'_1(t),\dots,S'_m(t) \imp S(t),\Delta$\\ 
with $t$ shallow and linear, and $\vars(t) \cap \vars(\Delta) = \emptyset$. 
\end{lemma}

\begin{proof}
	Let $(C;\pi) \in N$. 
	$(C;\pi)$ can be separated into variable disjoint components $(\Gamma_1,\ldots,\Gamma_n \imp \Delta_1,\ldots,\Delta_n; \pi_1 \wedge \ldots \wedge \pi_n)$, 
	where $\vert \Delta_i \vert \leq 1$ and $\lvar(\pi_i) \subseteq \vars(\Gamma_i \imp \Delta_i)$.
	For each positive literal $P(s) \in \Delta$ there is a fragment 
	\begin{align*}
		(A)~~ & (S_1(x_1),\dots,S_k(x_k) \imp P(s);\pi') \\
		\intertext{with $\{x_1,\dots,x_k\}\subseteq \vars(s)$. If $m > 0$, there is another fragment} 
		(B)~~ & (S_1(x_1),\dots,S_k(x_k),S'_1(t),\dots,S'_m(t) \imp ;\pi') \\
		\intertext{ or }
		(C)~~ & (S_1(x_1),\dots,S_k(x_k),S'_1(t),\dots,S'_m(t) \imp S(t);\pi') \\
		\intertext{ with $\{x_1,\dots,x_k\}\subseteq \vars(t)$, respectively.
			Lastly, for each variable $x\in \vars(C)$ with $x \notin \vars(t) \cup \vars(\Delta)$ there is a fragment  }
		(D)~~ & (S_1(x),\dots,S_k(x)\imp;\pi'). 
	\end{align*}
	
	Since there are only finitely many terms $s$ with depth($s$)$\leq d$ modulo renaming,
	there are only finitely many atomic constraints $x \neq s$ for a given variable $x$ different up to renaming $s$.
	Thus, a normal constraint can only contain finitely many combinations of subconstraints $\bwedge{i\in \Ind} ~ x \neq s_i$ without some $s_i$ being an instance of another $s_j$.  
	Therefore, for a fixed set of variables $x_1,\dots,x_k$, there are only finitely many constraints $\pi=\bwedge{i\in \Ind} ~~ z_i \neq s_i$ with $\lvar(\pi)\subseteq \{x_1,\dots,x_k\}$ up to variants.
	
	Since the number of predicates, function symbols, and their ranks is finite,
	the number of possible shallow and linear atoms $S(t)$ different up to variants is finite. 
	For a given shallow and linear $t$, there exist only finitely many clauses of the form $({S_{1}(t),\dots,S_{n}(t) \imp S(t)}; \pi)$ or $(S_{1}(t),\dots,S_{n}(t) \imp;\pi)$ with $\lvar(\pi)\subseteq \vars(t)$ modulo condensation and variants. 
	For a fixed set of variables $x_1,\dots,x_k$, there exist only finitely many clauses of the form $(S_1(y_1),\dots,S_k(y_l) \imp; \pi ) $ with $\{y_1,\dots,y_l\} \cup \lvar(\pi)\subseteq \{x_1,\dots,x_k\}$ modulo condensation and variants. 
	Therefore,  there are only finitely many distinct clauses of each form (A)-(D) without variants or condensations.
	
	If in the clause $(C;\pi)=(\Gamma_1,\ldots,\Gamma_n \imp \Delta_1,\ldots,\Delta_n; \pi_1 \wedge \ldots \wedge \pi_n)$ for some $i\neq j$, $(\Gamma_i \imp \Delta_i;\pi_i)$ is a variant of $(\Gamma_j \imp \Delta_j;\pi_j)$,
	then $(C;\pi)$ has a condensation and is therefore not part of $N$.
	Hence, there can be only finitely many different $(C;\pi)$  without variants or condensations and thus $N$ is finite. \qed
\end{proof}

\begin{lemma}[Finite Saturation]\label{lem:decide:saturation}
 Let $N$ be an MSL(SDC) clause set. 
Then $N$ can be finitely saturated up to redundancy by SDC-resolution with selection function $\mathrm{sel}$. 
\end{lemma}

\begin{proof}
	The general idea is that given the way $\sel$ is defined the clauses involved in constrained resolution and factoring can only fall into certain patterns.
	Any result of such inferences then is either strictly smaller than one of its parents by some terminating measure
	or falls into a set of clauses that  is bounded by Lemma~\ref{lem:decide:finite_base}.
	Thus, there can be only finitely many inferences before $N$ is saturated.
	
	Let $d$ be an upper bound on the depth of constraints found in $N$ and $\Sigma$ be the finite signature consisting of the function and predicate symbols occurring in $N$.
	Let $(\Gamma_1 \imp \Delta_1, S(t);\pi_1)$ and $(\Gamma_2 , S(t') \imp \Delta_2;\pi_2)$ be clauses in $N$ where sdc-resolution applies with $\sigma=\mgu(S(t),S(t'))$ and resolvent $R=((\Gamma_1,\Gamma_2\imp \Delta_1,\Delta_2)\sigma;\norm{(\pi_1 \wedge \pi_2)\sigma})$.
	
	Because no literal is selected by $\sel$, $\Gamma_1 \imp \Delta_1, S(t)$ can match only one of two patterns:
	\begin{align*}
		(A)~~ &  S_1(x_1),\dots, S_n(x_n) \imp S(f(y_1,\dots,y_k)),\Delta 
		\intertext{where $t=f(y_1,\dots,y_k)$ and $\{x_1,\dots,x_n\}\subseteq\{y_1,\dots,y_k\}\cup \vars(\Delta)$.} 
		(B) ~~ & S_1(x_1),\dots, S_n(x_n) \imp S(y),\Delta 
		\intertext{where $t=y$ and $x_1,\dots,x_n$ are variables in $\vars(\Delta)$, i.e., $y$ occurs only once.}
	\end{align*}
	The literal $S(t')$ is selected by $\sel$ in $\Gamma_2 , S(t') \imp \Delta_2$, and therefore $\Gamma_2 , S(t') \imp \Delta_2$ can match only one of the following three patterns:
	\begin{align*}
		(1)~~ &S(f(t_1,\dots,t_k)),\Gamma' \imp \Delta' \\
		(2)~~ & S(y'),\Gamma' \imp \Delta'  \text{ where $\Gamma'$ has no function terms and $y\notin \vars(\Delta')$. }\\
		(3)~~ & S(y'),\Gamma' \imp S'(y'),\Delta' \text{ where $\Gamma'$ has no function terms. }
	\end{align*}
	This means that the clausal part $(\Gamma_1,\Gamma_2\imp \Delta_1,\Delta_2)\sigma$ of $R$ has one of  six forms:
	\begin{align*}
		(A1)~~&  S_1(x_1)\sigma,\dots, S_n(x_n)\sigma, \Gamma' \imp \Delta,\Delta' \text{ with $\sigma=\{{y_1 \mapsto t_1},\dots \}$.}\\
		\intertext{ $\Delta\sigma = \Delta$ because  $S(f(y_1,\dots,y_k))$ and $\Delta$ do not share variables.} 
		(B1)~~&  S_1(x_1),\dots, S_n(x_n), \Gamma' \imp \Delta,\Delta'.\\
		\intertext {The substitution $\{y \mapsto f(t_1,\dots,t_k) \}$ is irrelevant since $S(y)$ is the only literal with variable $y$.}
		(A2)~~&  S_1(x_1),\dots, S_n(x_n), \Gamma'\tau \imp \Delta,\Delta' \text{ with  $\tau=\{{y' \mapsto f(y_1,\dots,y_k)} \}$.}\\
		\intertext{ $\Delta'\tau=\Delta'$ because $y' \notin \vars(\Delta')$.}
		(B2)~~&  S_1(x_1),\dots, S_n(x_n), \Gamma' \imp \Delta,\Delta'. \\
		(A3)~~&  S_1(x_1),\dots, S_n(x_n), \Gamma'\tau \imp S'(f(y_1,\dots,y_k)),\Delta,\Delta' \text{ with  $\tau=\{y \mapsto f(y_1,\dots,y_k) \}$.}\\
		\intertext{ $\Delta'\tau=\Delta'$ because $y' \notin \vars(\Delta')$.}
		(B3)~~&  S_1(x_1),\dots, S_n(x_n), \Gamma' \imp S'(y'),\Delta,\Delta'.
	\end{align*}
	
	In the constraint $\norm{(\pi_1 \wedge \pi_2) \sigma}$  the maximal depth of the subconstraints is less or equal to the maximal depth of $\pi_1$ or $\pi_2$.
	Hence, $d$ is also an upper bound on the constraint of the resolvent.
	In each case, the resolvent is again an MSL(SDC) clause.
	
	In the first and second case, the multiset of term depths of the negative literals in $R$ is strictly smaller than for the right parent.
	In both, the $\Gamma$ is the same between the right parent and the resolvent. 
	Only the $f(t_1,\dots,t_k)$ term is replaced by $x_1\sigma,\dots, x_n\sigma$ and $x_1,\dots, x_n$ respectively.
	In the first case, the depth of the $x_i\sigma$ is either zero if $x_i\notin \{y_1,\dots,y_k\}$ or at least one less than $f(t_1,\dots,t_k)$ since $x_i\sigma=t_i$.
	In the second case, the $x_i$ have depth zero which is strictly smaller than the depth of $f(t_1,\dots,t_k)$.
	Since the multiset ordering on natural numbers is terminating, the first and second case  can only be applied finitely many times by constrained resolution.
	
	In the third to sixth case $R$ is of the form
	$(S_1(x_1),\dots,S_l(x_l),S'_1(t),\dots,S'_m(t) \imp \Delta;\pi)$ or 
	$(S_1(x_1),\dots,S_l(x_l),S'_1(t),\dots,S'_m(t) \imp S(t)),\Delta;\pi)$ with $t=f(y_1,\dots,y_k)$.
	By Lemma \ref{lem:decide:finite_base}, there are only finitely many such clauses after condensation and removal of variants.
	Therefore, these four cases can apply only finitely many times during saturation.
	
	Let $(\Gamma \imp \Delta, S(t), S(t');\pi)$ be a clause in $N$ where sdc-factoring applies with $\sigma=\mgu(S(t),S(t'))$ and $R=((\Gamma \imp \Delta, S(t))\sigma;\norm{\pi\sigma})$.
	Because in $\Gamma \imp \Delta, S(t), S(t')$  no literal is selected, $\Gamma \imp \Delta, S(t), S(t')$ and $(\Gamma \imp \Delta, S(t))\sigma$ can only match  one of three patterns.
	\begin{align*}
		(A)~~& S_1(x_1),\dots, S_n(x_n) \imp S(f(y_1,\dots,y_k)),S(f(z_1,\dots,z_l)),\Delta \\
		\intertext {where $t=f(y_1,\dots,y_k)$, $t'=f(z_1,\dots,z_k)$, and $\{x_1,\dots,x_n\}\subseteq\{y_1,\dots,y_k\}\cup\{z_1,\dots,z_l\}\cup \vars(\Delta)$. The result is} 
		& S_1(x_1)\sigma,\dots, S_n(x_n)\sigma \imp S(f(y_1,\dots,y_k)),\Delta  \text{ with $\sigma=\{{z_1 \mapsto y_1},\dots \}$.}\\
		(B)~~&  S_1(x_1),\dots, S_n(x_n) \imp S(f(y_1,\dots,y_k)),S(z),\Delta \\
		\intertext{ where $t=f(y_1,\dots,y_k)$, $t'=z$ and  $\{x_1,\dots,x_n\}\subseteq\{y_1,\dots,y_k\}\cup\vars(\Delta)$, i.e., $z$ occurs only once. The result is}
		& S_1(x_1),\dots, S_n(x_n) \imp S(f(y_1,\dots,y_k)),\Delta.\\
		(C)~~& S_1(x_1),\dots, S_n(x_n) \imp S(y),S(z),\Delta \\
		\intertext{ where $t=y$, $t'=z$ and  $\{x_1,\dots,x_n\}\subseteq\vars(\Delta)$, i.e., $y$ and $z$ occur only once. The result is } 
		& S_1(x_1),\dots, S_n(x_n) \imp S(y),\Delta.
	\end{align*}
	
	In the new constraint $\norm{\pi \sigma}$ the maximal depth of the subconstraints is less or equal to the maximal depth of $\pi$.
	Hence $d$ is also an upper bound on the constraint of the resolvent.
	In each case, the resolvent is again  an MSL(SDC) clause.
	
	Furthermore, in each case the clause is of the form $(S_1(x_1),\dots,S_l(x_l) \imp \Delta; \pi)$. 
	By Lemma \ref{lem:decide:finite_base}, there are only finitely many such clauses  after condensation and removal of variants.
	Therefore, these three cases can  apply only finitely many times during saturation.\qed
\end{proof}

\begin{theorem}[MSL(SDC) Decidability]\label{theo:decide:main}
Satisfiability of the MSL(SDC) first-order fragment is decidable. 
\end{theorem}

\begin{proof}
	Follows from Lemma \ref{lem:decide:saturation} and \ref{lem:decide:complete}.
\end{proof}

\section{Approximation and Refinement}\label{sec:approx}

In the following, we show how decidability of the MSL(SDC) fragment can be used
to improve the approximation refinement calculus presented in \cite{Teucke2015}.

Our approach is based on a counter-example guided abstraction refinement (CEGAR) idea.
The procedure loops trough four steps: approximation, testing (un)satisfiability, lifting, and refinement.
The approximation step transforms any first-order logic clause set
into the decidable MSL(SDC) fragment while preserving unsatisfiability.
The second step employs the decidability result for  MSL(SDC), Section~\ref{sec:decide}, 
to test satisfiability of the approximated clause set.
If the approximation is satisfiable, the original problem is satisfiable as well and we are done.
Otherwise, the third step, lifting, tests whether the proof of unsatisfiability found for the approximated
clause set can be lifted to a proof of the original clause set.
If so, the original clause set is unsatisfiable and we are again done.
If not, we extract a cause for the lifting failure that always amounts to two different
 instantiations of the same variable in a clause from the original clause set. 
This is resolved by the fourth step, the refinement. The crucial clause in the
original problem is replaced and instantiated in a satisfiability preserving way 
such that the different instantiations do not reoccur anymore in subsequent iterations of the loop. 

As mentioned before, our motivation to use dismatching constraints is that for an unconstrained clause
the refinement adds quadratically many new clauses to the clause set.
In contrast, with constrained clauses the same can be accomplished with adding just a single new clause.
This extension is rather simple as constraints are treated the same as the antecedent literals in the clause. 
Furthermore we present refinement as a separate transformation rule. 

The second change compared to the previous version is the removal of the Horn approximation rule,
where we have now shown in Section~\ref{sec:decide} that a restriction to Horn clauses is not required for decidability anymore.
Instead, the linear and shallow approximations are extended to apply to non-Horn clauses instead.

The approximation consists of individual transformation rules $N \Rightarrow N'$ that are non-deterministically applied. 
They transform a clause that is not in the MSL(SDC) fragment
in finite steps into MSL(SDC) clauses. Each specific property of MSL(SDC) clauses, i.e, monadic predicates,
shallow and linear positive literals, is generated by a corresponding rule: the Monadic transformation encodes non-Monadic predicates as functions,
the shallow transformation extracts non-shallow subterms by introducing fresh predicates and the linear transformation renames 
non-linear variable occurrences. 

Starting from a constrained clause set $N$ the transformation is parameterized by a single monadic projection predicate $T$, 
fresh to $N$ and for each non-monadic predicate $P$ a separate projection function $f_P$ fresh to $N$. 
The clauses in $N$ are called the original clauses while the clauses in $N'$ are  the approximated clauses. 
We assume all clauses in $N$ to be variable disjoint.

\begin{definition}\label{def:approx:termencod}
Given a predicate $P$, projection predicate $T$, and projection function $f_P$,
define the injective function $\Proj{P}^T(P(\args{t})) := T(f_p(\args{t}))$
and $\Proj{P}^T(Q(\args{s})) := Q(\args{s})$ for $P \neq Q$.
The function is extended to [constrained] clauses, clause sets and interpretations.
Given a signature $\Sigma$ with non-monadic predicates $P_1,\ldots,P_n$, define $\Proj{\Sigma}^T(N):=\Proj{P_1}^T(\ldots(\Proj{P_n}^T(N))\ldots)$ and $\Proj{\Sigma}^T(\I{}):=\Proj{P_1}^T(\ldots(\Proj{P_n}^T(\I{}))\ldots)$.
\end{definition}

\bigskip
\shortrules{Monadic}{$N$}{$\Proj{P}^T(N)$}{provided $P$ is a non-monadic predicate in the signature of $N$.}{MO}{15}

\bigskip
\shortrules{Shallow}{$N~\dot{\cup}~\{(\Gamma \imp E[s]_{p},\Delta;\pi)\}$}{\\ $~~~~~~~~~~~~~~~~~~~~~~~ N\cup\{(S(x),\Gamma_l \imp E[p/x],\Delta_l;\pi)$; $(\Gamma_r \imp S(s),\Delta_r;\pi)\}$}
{provided $s$ is complex, $\vert p\vert=2$, $x$ and $S$ fresh,  $\Gamma_l\{x \mapsto s\} \cup \Gamma_r = \Gamma$, $\Delta_l {\cup} \Delta_r = \Delta$, 
$\{Q(y)\in \Gamma \mid {y \in \vars(E[p/x],\Delta_l)\}} \subseteq \Gamma_l$, 
$\{Q(y)\in \Gamma  \mid {y \in \vars(s,\Delta_r) \}} \subseteq \Gamma_r$.}{SH}{15}

\bigskip
\shortrules{Linear 1}{$N~\dot{\cup}~\{(\Gamma \imp \Delta, E'[x]_{p},E[x]_q;\pi)\}$}{\\ $~~~~~~~~~~~~~~~~~~~~~~\;N\cup\{(\Gamma\sigma,\Gamma \imp \Delta, E'[x]_{p},E[q/x'];\pi \wedge \pi\sigma)\}$}
{provided $x'$ is fresh and $\sigma= \{x \mapsto x'\}$.}{LI}{15}

\bigskip
\shortrules{Linear 2}{$N~\dot{\cup}~\{(\Gamma \imp \Delta, E[x]_{p,q};\pi)\}$}{\\ $~~~~~~~~~~~~~~~~~~~~~~\;N\cup\{(\Gamma\sigma,\Gamma \imp \Delta, E[q/x'];\pi \wedge \pi\sigma)\}$}
{provided $x'$ is fresh, $p \neq q$ and $\sigma= \{x \mapsto x'\}$.}{LI}{15}
\bigskip

\shortrules{Refinement}{$N~\dot{\cup}~\{(C,\pi)\} $} {$N \cup  \{(C;\pi \wedge x \neq t),(C;\pi)\{x \mapsto t\}\}$}
{provided $x\in \vars(C)$, $t$ straight and $\vars(t) \cap \vars((C,\pi))=\emptyset$. }{\Refine}{15}
\bigskip

Note that variables are not renamed unless explicitly stated in the rule.
This means that original clauses and their approximated counterparts share variable names.
We use this to trace the origin of variables in the approximation.

The refinement transformation $\refine$ is not needed to eventually generate MSL(SDC) clauses, but can be 
used to achieve a  more fine-grained approximation of $N$, see below.

In the shallow transformation, $\Gamma$ and $\Delta$ are separated into 
$\Gamma_l$, $\Gamma_r$, $\Delta_l$, and $\Delta_r$, respectively.
The separation can be almost arbitrarily chosen as long as no atom from $\Gamma$, $\Delta$ is skipped.
However, the goal is to minimize the set of shared variables, i.e., the variables of $(\Gamma \imp E[s]_{p},\Delta;\pi)$ that are inherited by both approximation clauses, $\vars(\Gamma_r,s,\Delta_r) \cap \vars(\Gamma_l,E[p/x],\Delta_l)$.
If there are no  shared variables, the shallow transformation is satisfiability equivalent.
The conditions on $\Gamma_l$ and $\Gamma_r$ ensure that $S(x)$ atoms are not separated from the respective positive occurrence of  $x$ in subsequent shallow transformation applications.

Consider the clause $Q(f(x),y) \imp P(g(f(x),y))$. 
The simple shallow transformation $S(x'),Q(f(x),y) \imp P(g(x',y)); S(f(x))$ is not satisfiability equivalent -- nor with any alternative partitioning of $\Gamma$.
However, by replacing the occurrence of the extraction term $f(x)$ in $Q(f(x),y)$ with the fresh variable $x'$, 
the approximation $S(x'),Q(x',y) \imp P(g(x',y)); S(f(x))$ is satisfiability equivalent.
Therefore, we allow the extraction of $s$ from the terms in $\Gamma_l$ and require  $\Gamma_l\{x \mapsto s\} \cup \Gamma_r = \Gamma$.
  
We consider Linear~1 and Linear~2 as two cases of the same linear transformation rule.
Their only difference is whether the two occurrences of $x$ are in the same literal or not. 
The duplication of literals and constraints in $\Gamma$ and $\pi$ is not needed if $x$ does not occur in $\Gamma$ or $\pi$.

Further, consider a linear transformation $N \cup \{(C;\pi)\} \li  N \cup \{(C_a;\pi_a)\}$,
where a fresh variable $x'$ replaces an occurrence of a non-linear variable $x$ in $(C;\pi)$.  
Then, $(C_a;\pi_a)\{x' \mapsto x\}$ is equal to $(C;\pi)$ modulo duplicate literal elimination.
A similar property can be observed of a resolvent of $(C_l;\pi)$ and $(C_r;\pi)$ 
resulting from a shallow transformation $N \cup \{(C;\pi)\} \sh  N \cup \{(C_l;\pi), (C_r;\pi)\}$.
Note that by construction, $(C_l;\pi)$ and $(C_r;\pi)$ are not necessarily variable disjoint.
To simulate standard resolution, we need to rename at least the shared variables in one of them. 
  
\begin{definition}[$\apr$] \label{def:approx:apr}
We define $\apr$ as the priority rewrite system~\cite{Baeten89tcs} 
consisting of $\refine$, $\Rightarrow_{\operatorname{MO}}$, $\Rightarrow_{\operatorname{SH}}$  and  $\Rightarrow_{\operatorname{LI}}$ with 
priority $\refine \,>\, \Rightarrow_{\operatorname{MO}}  \,>\, \Rightarrow_{\operatorname{SH}} \,>\, \Rightarrow_{\operatorname{LI}}$,
where $\refine$ is only applied finitely many times.
\end{definition}

\begin{lemma}[$\apr$ is a Terminating Over-Approximation]\label{lem:approx:sound}
(i)~$\apr^*$ terminates,
(ii) if $N \apr N'$ and $N'$ is satisfiable, then $N$ is also satisfiable.
\end{lemma}

\begin{proof}
	(i)~The transformations can be considered sequentially, because of the imposed rule priority. 
	There are, by definition, only finitely many refinements at the beginning of an approximation $\apr^*$.
	The monadic transformation strictly reduces the number of non-monadic atoms.
	The shallow transformation strictly reduces the multiset of term depths of the newly introduced clauses compared
	to the removed parent clause.
	The linear transformation strictly reduces the number of duplicate variable occurrences in positive literals. 
	Hence $\apr$ terminates.
	
	(ii)  Let $N \cup \{ (C;\pi)\} \li N \cup \{ (C_a;\pi_a)\}$ where an occurrence of a variable $x$ in $(C;\pi)$ is replaced by a fresh $x'$.
	As $(C_a;\pi_a)\{x' \mapsto x\}$ is equal to $(C;\pi)$ modulo duplicate literal elimination,
	$\I{} \models (C;\pi)$ if $\I{} \models (C_a;\pi_a)$.
	Therefore, the linear transformation is an over-approximation.
	
	 Let $N \cup \{ (C;\pi)\} \sh N \cup \{ (C_l;\pi_l),(C_r;\pi_r)\}$ 
	and  $(C_a;\pi_a)$ 
	be the shallow $\rho$-resolvent.
	As  $(C_a;\pi_a)\rho^{-1}$  equals $(C;\pi)$ modulo duplicate literal elimination,
	$\I{} \models (C;\pi)$ if $\I{} \models (C_l;\pi_l), (C_r;\pi_r)$.
	Therefore, the shallow transformation is an over-approximation.
	
	 Let $N \mo \Proj{P}(N)=N'$.
	Then, $N=\R{P}(N') $.
	Let $\I{}$ be a model of $N'$ and $(C;\pi) \in N$.  
	Since $ \Proj{P}((C;\pi)) \in N'$ , $\I{} \sat \Proj{P}((C;\pi))$ and thus, $\R{P}(\I{})\sat (C;\pi)$.
	Hence, $\R{P}(\I{})$ is a  model of $N$.
	Therefore, the monadic transformation is an over-approximation. Actually, it
	is a satisfiability preserving transformation.
	
	  Let $N \cup \{(C;\pi)\}\refine N \cup \{(C;\pi \wedge x \neq t),(C;\pi)\{x \mapsto t\}\}$.
	Let $C\delta \in \ground{(C;\pi)} $. 
	If $x\delta$ is not an instance of $t$, then $\delta$ is a solution of $\pi \wedge x \neq t$ and  $C\delta \in \ground{(C;\pi \wedge x \neq t)}$.
	Otherwise, $\delta=\{x\mapsto t\}\delta'$ for some substitution $\delta'$.
	Then, $\delta$ is a solution of $\pi\{x\mapsto t\}$ and thus, $C\delta=C\{x\mapsto t\}\delta' \in \ground{(C\{x \mapsto t\};\pi\{x \mapsto t\})}$.
	Hence, $ \ground{(C;\pi)} \subseteq  \ground{(C;\pi \wedge x \neq t)} \cup  \ground{(C;\pi)\{x \mapsto t\}}.$
	Therefore, if $\I{}$ is  a model of $N \cup \{(C;\pi \wedge x \neq t),(C;\pi)\{x \mapsto t\}\}$, then   $\I{}$ is also a model of $N \cup \{(C;\pi)\}$.  \qed
\end{proof}

Note that $\refine$ and $\mo$  are also satisfiability preserving transformations.

\begin{corollary}\label{cor:approx:sound}
If $N \apr^* N'$ and $N'$ is satisfied by a model $\I{}$, 
then  $\R{\Sigma}(\I{})$  is a model of $N$.
\end{corollary}
\begin{proof}
	Follows from Lemma~\ref{lem:approx:sound} (ii)-(v).\qed
\end{proof}

On the basis of $\apr$ we can define an ancestor relation $\anc$ that relates
clauses, literal occurrences, and variables  with respect to approximation. 
This relation is needed in order to figure out the exact clause, literal, variable for
refinement.

\begin{definition}[The Shallow Resolvent]\label{def:approx:resolvent}
	Let $N \cup \{(C;\pi)\} \sh  N \cup \{(C_l;\pi), (C_r;\pi)\}$ with $C=\Gamma \imp E[s]_{p},\Delta$, $C_l=S(x),\Gamma_l \imp E[p/x],\Delta_l$ and $C_r= \Gamma_r \imp S(s),\Delta_r$.
	Let $x_1,\ldots,x_n$ be the variables shared between $C_l$ and $C_r$ and $\rho=\{x_1 \mapsto x'_1, \ldots, x_n \mapsto x'_n\}$ be a variable renaming with $x'_1,\ldots,x'_n$ fresh in $C_l$ and $C_r$.
	We define $(\Gamma_l\{x \mapsto s\rho \},\Gamma_r\rho \imp E[p/s\rho],\Delta_l,\Delta_r\rho;\pi \wedge \pi\rho)$ as the shallow $\rho$-resolvent.
\end{definition}

Let $(C_a;\pi_a)$ be the shallow  $\rho$-resolvent of $N \cup \{(C;\pi)\} \sh  N \cup \{(C_l;\pi), (C_r;\pi)\}$.
Note that for any two ground instances $C_l\delta_l$ and $C_r\delta_r$, their resolvent is a ground instance of $(C_a;\pi_a)$.
Furthermore, using the reverse substitution $\rho^{-1}= \{x'_1 \mapsto x_1, \ldots, x'_n \mapsto x_n\}$, 
$(C_a;\pi_a)\rho^{-1}= (\Gamma_l\{x \mapsto s \},\Gamma_r \imp E[s]_{p},\Delta_l,\Delta_r;\pi \wedge \pi)$ is equal to $(C;\pi)$ modulo duplicate literal elimination. 
This is because,  $\Delta_l \cup \Delta_r = \Delta$ and $\Gamma_l\{x \mapsto s \} \cup \Gamma_r = \Gamma$  by definition of $\sh$ and $\pi \wedge \pi$ is equivalent to $ \pi$.

Next, we establish parent relations that link original and approximated clauses, as well as their variables and literals. 
Together the parent, variable and literal relations will allow us to not only trace  any approximated clause back to their origin,
but also predict what consequences changes to the original set will have on its approximations.

For the following definitions, we assume that clause and literal sets are lists and that $\Proj{P}^T$ and substitutions act as mappings.
This means we can uniquely identify clauses and literals by their position in those lists.
Further, for every shallow transformation $N \sh  N'$,
we will also include the shallow resolvent in the parent relation as if it were a member of $N'$.

\begin{definition}[Parent Clause]\label{def:approx:pclause}
	For an approximation step $N \apr N'$ and two clauses $(C;\pi)\in N$ and $(C';\pi')\in N'$,
	we define $[(C;\pi), N] \anc [(C';\pi'), N']$ expressing that $(C;\pi)$ in $N$ is the parent clause of $(C';\pi')$ in $N'$:\\
	If $N \mo \Proj{P}^T(N)$, then 
	
	$[(C;\pi),N] \anc [\Proj{P}^T((C;\pi)), \Proj{P}^T(N)]$ for all $(C;\pi) \in N$.\\
	If $N=N'' \cup \{(C;\pi)\} \sh N'' \cup \{(C_l;\pi_l),(C_r;\pi_r)\}=N'$, then 
	
	$[(D,\pi'),N] \anc [(D,\pi'),N']$ for all $(D,\pi') \in N''$ and 
	
	$[(C,\pi),N] \anc [(C_l;\pi_l),N']$,
	
	$[(C,\pi),N] \anc [(C_r;\pi_r),N']$ and 
	
	$[(C,\pi),N] \anc [(C_a;\pi_a),N']$ for any shallow resolvent  $(C_a;\pi_a)$.\\
	If $N=N'' \cup \{(C;\pi)\} \li N'' \cup \{(C_a;\pi_a)\}=N'$, then 
	
	$[(D,\pi'),N] \anc [(D,\pi'),N']$ for all $(D,\pi') \in N''$ and 
	
	$[(C,\pi),N] \anc [(C_a,\pi_a),N']$. \\
	If $N=N'' \cup \{(C;\pi)\}\refine N'' \cup \{(C;\pi \wedge x \neq t),(C;\pi)\{x \mapsto t\}\}=N'$, then 
	
	$[(D,\pi'),N] \anc [(D,\pi'),N']$ for all $(D,\pi') \in N''$ ,
	
	$[(C,\pi),N] \anc [(C;\pi \wedge x \neq t),N']$ and 
	
	$[(C,\pi),N] \anc [(C;\pi)\{x \mapsto t\}, N'] $.
\end{definition}

\begin{definition}[Parent Variable]\label{def:approx:pvar}
	Let $N \apr N'$ be an approximation step and $[(C;\pi),N] \anc [(C';\pi'),N']$.
	For two variables $x$ and $y$,
	we define $[x,(C;\pi), N] \anc [y,(C';\pi'), N']$ expressing that $x \in \vars(C)$ is the parent variable of $y \in \vars(C')$:\\
	If $x\in \vars((C;\pi))\cap \vars((C';\pi'))$, then 
	
	$[x,(C;\pi),N] \anc [x,(C';\pi'),N']$.\\
	If $N \sh N'$ and  $(C',\pi')$ is the shallow $\rho$-resolvent, 
	
	$[x_i,(C;\pi),N] \anc [x_i\rho,(C';\pi'),N']$ for each $x_i$ in the domain of $\rho$.\\
	If $N \li N'$, $C= \Gamma \imp \Delta[x]_{p,q}$ and  $C'=\Gamma\{x \mapsto x'\},\Gamma \imp \Delta[q/x']$, then 
	
	$[x,(C;\pi),N] \anc [x',(C';\pi'),N']$.
\end{definition}

Note that if $N \sh N'$ and $x$ is the fresh extraction variable in $(C_l;\pi_l)$, then $x$  has no parent variable.
For literals, we actually further specify the relation on the positions within literals of a clause $(C;\pi)$ using pairs $(L,r)$ of literals and positions. 
We write $(L,r)\in C$ to denote that $(L,r)$ is a literal position in $(C;\pi)$ if $L\in C$ and $r\in \pos(L)$.  
Note that a literal position $(L,r)$ in $(C;\pi)$ corresponds to the term $L\vert_r$.

\begin{definition}[Parent literal position]\label{def:approx:pterm}
	Let $N \apr N'$ be an approximation step and $[(C;\pi),N] \anc [(C';\pi'),N']$.
	For two literal positions $(L,r)$ and $(L',r')$,
	we define $[r,L,(C;\pi), N] \anc [r',L',(C';\pi'), N']$ expressing that $(L,r)$ in $(C;\pi)$ is the parent literal position of $(L',r')$ in $(C';\pi')$:\\
	If $(C;\pi)=(C';\pi')$, then 
	
	$[r,L,(C;\pi),N] \anc [ r,L,(C';\pi'),N']$ for all $(L,r)\in C$.\\
	If $N \refine N'$ and $(C',\pi')=(C;\pi \wedge x \neq t)$,  then 
	
	$[r,L,(C;\pi),N ] \anc [ r,L,(C';\pi'),N']$ for all $(L,r)\in C$.\\
	If $N \refine N'$ and $(C',\pi')=(C;\pi)\{x \mapsto t\}$,  then 
	
	$[r,L,(C;\pi),N ] \anc [ r,L\{x \mapsto t\},(C';\pi'),N']$ for all $(L,r)\in C$.\\
	If $N \mo \Proj{P}^T(N)=N'$, then 
	
	$[\varepsilon,P(\args{t}),(C;\pi),N] \anc [ \varepsilon,T(f_p(\args{t})),(C';\pi'),N']$ for all $P(\args{t})\in C$ and
	
	$[r,P(\args{t}),(C;\pi),N] \anc [ 1.r,T(f_p(\args{t})),(C';\pi'),N']$ for all $(P(\args{t}),r)\in C$.\\
	If $N \sh N'$, $C= \Gamma \imp E[s]_{p},\Delta$ and $C'=S(x),\Gamma_l \imp E[p/x],\Delta_l$, then 
	
	$[r,E[s]_{p},(C;\pi),N] \anc [ r,E[p/x],(C';\pi'),N']$ for all $r\in \pos(E[p/x])$,  
	
	$[p,E[s]_{p},(C;\pi),N] \anc [ r,S(x),(C';\pi'),N']$ for all $r\in \pos(S(x))$,
	
	$[r,L\{x \mapsto s\},(C;\pi),N] \anc [ r,L,(C';\pi'),N']$ for all $(L,r)\in \Gamma_l$,  
	
	$[r,L,(C;\pi),N] \anc [ r,L,(C';\pi'),N']$ for all $(L,r)\in \Delta_l$.\\
	If $N \sh N'$, $C= \Gamma \imp E[s]_{p},\Delta$ and $C'=\Gamma_r \imp S(s),\Delta_r$, then 
	
	$[p,E[s]_{p},(C;\pi),N] \anc [ \varepsilon ,S(s),(C';\pi'),N']$,
	
	$[pr,E[s]_{p},(C;\pi),N] \anc [ 1.r,S(s),(C';\pi'),N']$ for all $r\in \pos(s)$, and
	
	$[r,L,(C;\pi),N] \anc [ r,L,(C';\pi'),N']$ for all $(L,r)\in \Gamma_r\cup\Delta_r$.\\
	If $N \sh N'$, $C= \Gamma \imp E[s]_{p},\Delta$ and $(C',\pi')$ is the shallow $\rho$-resolvent, then 
	
	$[r,E[s]_{p},(C;\pi),N] \anc [ r,E[p/s\rho],(C';\pi'),N']$ for all $r\in \pos(E[p/s\rho])$,  
	
	$[r,L\{x \mapsto s\},(C;\pi),N] \anc [ r,L\{x \mapsto s\rho \},(C';\pi'),N']$ for all $(L,r)\in \Gamma_l$,  
	
	$[r,L,(C;\pi),N] \anc [ r,L\rho,(C';\pi'),N']$ for all $(L,r)\in \Gamma_r\cup\Delta_r$, and
	
	$[r,L,(C;\pi),N] \anc [ r,L,(C';\pi'),N']$ for all $(L,r)\in \Delta_l$.\\
	If $N \li N'$, $C= \Gamma \imp \Delta, E'[x]_{p},E[x]_q$ and $C'=\Gamma\{x \mapsto x'\},\Gamma \imp \Delta, E'[x]_{p},E[q/x']$,  
	
	$[r,E'[x]_{p},(C;\pi),N] \anc [ r,E'[x]_p,(C';\pi'),N']$ for all $r\in \pos(E'[x]_p)$, 
	
	$[r,E[x]_{q},(C;\pi),N] \anc [ r,E[q/x'],(C';\pi'),N']$ for all $r\in \pos(E[q/x'])$,,
	
	$[r,L,(C;\pi),N] \anc [ r,L\{x \mapsto x'\},(C';\pi'),N']$ for all $(L,r)\in \Gamma$, 
	
	$[r,L,(C;\pi),N] \anc [ r,L,(C';\pi'),N']$ for all $(L,r)\in \Gamma$, and 
	
	$[r,L,(C;\pi),N] \anc [ r,L,(C';\pi'),N']$ for all $(L,r)\in \Delta$.\\
	If $N \li N'$, $C= \Gamma \imp \Delta, E[x]_{p,q}$ and $C'=\Gamma\{x \mapsto x'\},\Gamma \imp \Delta, E[q/x']$,  then 
	
	$[r,E[x]_{p,q},(C;\pi),N] \anc [ r,E[q/x'],(C';\pi'),N']$ for all $r\in \pos(E[q/x'])$, 
	
	$[r,L,(C;\pi),N] \anc [ r,L\{x \mapsto x'\},(C';\pi'),N']$ for all $(L,r)\in \Gamma$, 
	
	$[r,L,(C;\pi),N] \anc [ r,L,(C';\pi'),N']$ for all $(L,r)\in \Gamma$, and 
	
	$[r,L,(C;\pi),N] \anc [ r,L,(C';\pi'),N']$  for all $(L,r)\in \Delta$.\\
\end{definition}

\begin{figure}[h]
	\begin{tikzpicture}
	\path (0,5.5) node(x1) {$\Gamma$} ++(0.5,0) node(x2) {$\rightarrow$}  ++(0.5,0) node(x3) {$E$} ++(0.45,0) node(x4) {$[s]_p,$} ++(0.5,0) node(x5) {$\Delta$}
	(0,4) node(y1) {$\Gamma_l,$} ++(0.62,0) node(y2) {$S(x)$} ++(0.5,0) node(y3) {$\rightarrow$}  ++(0.5,0) node(y4) {$E$} ++(0.6,0) node(y5) {$[p/x],$} ++(0.65,0) node(y6) {$\Delta_l$};
	\draw [->] (x1) edge (y1) (x3) edge (y4) (x4) edge (y2) (x5) edge (y6);
	
	\path (3.85,5.5) node(x1) {$\Gamma$} ++(0.5,0) node(x2) {$\rightarrow$}  ++(0.5,0) node(x3) {$E$} ++(0.45,0) node(x4) {$[s]_p,$} ++(0.5,0) node(x5) {$\Delta$}
	(3.85,4) node(y1) {$\Gamma_r$}  ++(0.5,0) node(y2) {$\rightarrow$} ++(0.62,0) node(y3) {$S(x),$} ++(0.65,0) node(y4) {$\Delta_r$};
	\draw [->] (x1) edge (y1) (x4) edge (y3)  (x5) edge (y4);
	
	\path (7.25,5.5) node(x1) {$\Gamma$} ++(0.5,0) node(x2) {$\rightarrow$}  ++(0.5,0) node(x3) {$E$} ++(0.45,0) node(x4) {$[s]_p,$} ++(0.5,0) node(x5) {$\Delta$}
	(7.25,4) node(y1) {$\Gamma_l\{x \mapsto s\rho\},$}  ++(1.25,0) node(y2) {$\Gamma_r\rho$} ++(0.5,0) node(y3) {$\rightarrow$} ++(0.8,0) node(y4) {$E[p/s\rho],$} ++(0.9,0) node(y5) {$\Delta_l,$} ++(0.6,0) node(y6) {$\Delta_r\rho$};
	\draw [->] (x1) edge (y1) (x1) edge (y2) (x3) edge (y4)  (x5) edge (y5) (x5) edge (y6);
	
	\path (1,3.4) node(x) {shallow left} ++(3.75,0) node(x) {shallow right} ++(4,0) node(x) {shallow resolvent};
	
	\path (1,2.5) node(x1) {$\Gamma$} ++(0.5,0) node(x2) {$\rightarrow$} ++(0.5,0) node(x3) {$\Delta,$} ++(0.75,0) node(x4) {$E'[x]_p,$} ++(1,0) node(x5) {$E[x]_q$} 
	(1,1) node(y1) {$\Gamma\sigma,$} ++(0.62,0) node(y2) {$\Gamma$} ++(0.5,0) node(y3) {$\rightarrow$}  ++(0.65,0) node(y4) {$\Delta,$} ++(0.75,0) node(y5) {$E'[x]_p,$} ++(1.1,0) node(y6) {$E[q/x']$} ;
	\draw [->] (x1) edge (y1) (x1) edge (y2) (x3) edge (y4) (x4) edge (y5) (x5) edge (y6);
	
	\path (6.5,2.5) node(x1) {$\Gamma$} ++(0.5,0) node(x2) {$\rightarrow$} ++(0.5,0) node(x3) {$\Delta,$} ++(0.8,0) node(x4) {$E[x]_{p,q}$}  
	(6.5,1) node(y1) {$\Gamma\sigma,$} ++(0.62,0) node(y2) {$\Gamma$} ++(0.5,0) node(y3) {$\rightarrow$}  ++(0.65,0) node(y4) {$\Delta,$} ++(0.8,0) node(y5) {$E[q/x']$} ;
	\draw [->] (x1) edge (y1) (x1) edge (y2) (x3) edge (y4) (x4) edge (y5);
	
	\path (3,0.4) node(x) {linear 1} ++(5,0) node(x) {linear 2};
	
	\end{tikzpicture}

	\caption{Visual representation of the parent literal position relation (Definition~\ref{def:approx:pterm})}
\end{figure}

The  transitive closures of each parent relation are called ancestor relations.

The over-approximation of a clause set $N$ can introduce resolution refutations 
that have no corresponding equivalent in $N$ which we consider a lifting failure.
Compared to our previous calculus~\cite{Teucke2015}, the lifting process is identical
with the exception that there is no case for the removed Horn transformation.
We only update the definition of conflicting cores to consider constrained clauses.

\begin{definition}[Conflicting Core]\label{def:lifting:core}
	A finite set of unconstrained clauses and a solvable constraint $(N^\bot;\pi)$ are a conflicting core if 
	$N^\bot\delta$ is unsatisfiable for all solutions $\delta$ of $\pi$ over $\vars(N^\bot)\cup \lvar(\pi)$. 
	A conflicting core $(N^\bot;\pi)$ is a conflicting core of the constrained clause set $N$ if for every $C\in N^\bot$ there is a clause $(C',\pi')\in N$
	such that $(C;\pi)$ is an instance of $(C';\pi')$ modulo duplicate literal elimination. 
	The clause $(C';\pi')$ is then called the instantiated clause of $(C;\pi)$ in $(N^\bot;\pi)$.
	We call $(N^\bot;\pi)$ complete if for every clause $C \in N^\bot$ and literal $L\in C$, there exists a clause $D\in N^\bot$ with $\overline L\in D$. 
\end{definition}

A conflicting core is a generalization of a ground unsatisfiability core that allows global variables to act as parameters. 
This enables more efficient lifting and refinement compared to a simple ground unsatisfiable core.
We show some examples at the end of this section.

We discuss the potential lifting failures and the corresponding refinements only for the linear and shallow case 
because lifting the satisfiability equivalent monadic and refinement transformations always succeeds.
To reiterate from our previous work:
in the linear case, there exists a clause in the conflicting core that is not an instance of the original clauses.
In the shallow case, there  exists a pair of clauses whose resolvent is not an instance of the original clauses.
We combine these two cases by introducing the notion of a lift-conflict. 
 
\begin{definition}[Conflict]\label{def:refine:conflict}
Let $N \cup \{(C,\pi)\} \li  N\cup \{(C_a,\pi_a)\}$ and $N^\bot$ be a complete  ground conflicting core of $N\cup \{(C_a,\pi_a)\}$.
We call a conflict clause $C_c \in N^\bot$ with the instantiated clause $(C_a,\pi_a)$ a lift-conflict if $C_c$ is not an instance of $(C,\pi)$ modulo duplicate literal elimination.
Then, $C_c$ is an instance of $(C_a,\pi_a)$, which we call the conflict clause of $C_c$.

Let $N \cup \{(C,\pi)\} \sh  N \cup \{(C_l,\pi_l),(C_r,\pi_r)\}$, $(C_a;\pi_a)$ be the shallow resolvent and $N^\bot$ be a complete ground conflicting core of $N\cup  \{(C_l,\pi_l),(C_r,\pi_r)\}$.
We call the resolvent $C_c$ of $C_l\delta_l\in N^\bot$ and $C_r\delta_r \in N^\bot$ a lift-conflict if $C_c$ is not an instance of $(C,\pi)$ modulo duplicate literal elimination. 
Then, $C_c$ is an instance of $(C_a;\pi_a)$, which we call the conflict clause of $C_c$.
\end{definition}
 
The goal of refinement is to instantiate the original parent clause in such a way that is both satisfiability equivalent and prevents the lift-conflict 
after approximation.  
Solving the refined approximation will then either necessarily produce a complete saturation or  a new refutation proof, 
because its conflicting core has to be different.
For this purpose, we use the refinement transformation to segment the original parent clause $(C;\pi)$ into two parts $(C;\pi \wedge x \neq t)$ and $(C;\pi)\{x \mapsto t\}$. 

For example, consider  $N$ and its linear transformation $N'$.\newline
\centerline{$\begin{array}{r@{\,\imp\,}lcr@{\,\imp\,}l}
	  &  P(x,x) &  \;\Rightarrow_{\operatorname{LI}}\; &  &P(x,x')\\
	P(a,b) & & \;\apr^0\; & P(a,b)&\\
	\end{array}$}
The ground conflicting core of $N'$ is\newline
\centerline{$\begin{array}{r@{\,\imp\,}l}
	& P(a,b) \\
	P(a,b)& \\
	\end{array}$}
Because  $P(a,b)$ is not an instance of $P(x,x)$, lifting fails. 
$P(a,b)$ is the lift-conflict. 
Specifically, $\{x\mapsto a\}$ and $\{x \mapsto b\}$ are conflicting substitutions for the parent variable $x$.
We pick $\{x\mapsto a\}$ to segment $P(x,x)$ into $ (P(x,x);x \neq a)$ and $P(x,x)\{x \mapsto a\}$.
Now, any descendant of $(P(x,x);x \neq a)$ cannot have $a$ at the position of the first $x$, and
any descendant of $P(x,x)\{x \mapsto a\}$ must have an $a$ at the position of the second $x$.
Thus, $P(a,b)$ is excluded in both cases and no longer appears as a lift-conflict.

To show that the lift-conflict will not reappear in the general case, we use that the conflict clause and its ancestors 
have strong ties between their term structures and constraints.  

\begin{definition}[Constrained Term Skeleton]\label{def:refine:termskel}
The constrained term skeleton of a term $t$ under constraint $\pi$, $\skt(t,\pi)$, is defined as 
the normal form of the following transformation: \newline
\centerline{$\begin{array}{c}
 (t[x]_{p,q} ;\pi)  \Rightarrow_{\skt} (t[q/x'] ;\pi \wedge \pi\{x \mapsto x'\}), \text{ where } p \neq q \text{ and $x'$ is fresh}. 
 \end{array}$}
\end{definition}

The constrained term skeleton of a term $t$ is essentially a linear version of $t$ where the restrictions on each variable position imposed by $\pi$ are preserved. 
For $(t,\pi)$ and a solution $\delta$ of $\pi$, $t\delta$  is called a ground instance of $(t,\pi)$.

\begin{lemma}\label{lem:refinement:ancestor_skel}
Let $N_0 \apr^*  N_k$, $(C_k;\pi_k)$ in $N$ with the ancestor clause $(C_0;\pi_0)\in N_0$ and $N^\bot_k$ be a complete ground conflicting core of $N_k$. 
Let $\delta$ be a solution of $\pi_k$ such that $C_k\delta$ is in $N^\bot_k$.
If $(L',q')$ is a literal position in $(C_k;\pi_k)$ with the ancestor $(L,q)$ in $(C_0,\pi_0)$, 
then (i) $L'\delta\vert_{q'}$ is an instance of $\skt(L\vert_q,\pi_0)$, 
(ii) $q=q'$ if $L$ and $L'$ have the same predicate, and 
(iii) if $L'\vert_{q'}=x$ and there exists an ancestor variable $y$ of $x$ in $(C_0,\pi_0)$, then $L\vert_{q}=y$.
\end{lemma}

\begin{proof}
	By induction on the length of the approximation $N_0 \apr^* N_k$.
	
	The base case $N_k=N_0$, is trivial.
	Let $N_0= N \cup \{ (C;\pi)\} \sh N \cup \{ (C_l;\pi_l),(C_r;\pi_r)\}=N_k$, $(C_k;\pi_k)$ be the shallow $\rho$-resolvent and $C_k\delta$ be the resolvent of two instances of $(C_l;\pi_l)$ and $(C_r;\pi_r)$ in $N^\bot_k$.
	Then, $(C_k;\pi_k)\rho^{-1}$ is equal to $ (C;\pi)$ modulo duplicate literal elimination.
	Thus, by definition $(L,q)=(L',q')\rho^{-1}$.
	Therefore, (i) $L'\delta\vert_{q'}$ is an instance of $\skt(L\vert_q,\pi_0)$, 
	(ii) $q=q'$ if $L$ and $L'$ have the same predicate, and 
	(iii) if $L'\vert_{q'}=x$ and there exists an ancestor variable $y$ of $x$ in $(C_0,\pi_0)$, then $L\vert_{q}=y$.
	
	Now, let  $ N_0 \apr N_{1} \apr^* N_k$.
	Since $(L',p)$ has an ancestor literal position in $(C_0,\pi_0)$,
	the ancestor clause of $(C_k;\pi_k)$ in $N_1$, $(C_1,\pi_1)$, contains the the ancestor literal position $(L_1,{q_1})$, which has $(L,q)$ as its parent literal position. 
	By the induction hypothesis on $N_{1} \apr^* N_k$, 
	(i) $L'\delta\vert_{q'}$ is an instance of $\skt(L_1\vert_{q_1},\pi_1)$, 
	(ii) $q_1=q'$ if $L_1$ and $L'$ have the same predicate, and 
	(iii) if $L'\vert_{q'}=x$ and  there is an ancestor variable $y_1$ of $x$ in $(C_1,\pi_1)$, then $L_1\vert_{q_1}=y_1$.
	
	Let $N_0= N \cup \{ (C;\pi)\} \refine N \cup \{ (C;\pi \wedge x \neq t),(C;\pi)\{x \mapsto t\}\}=N_1$.
	If $(C_1,\pi_1)$ is neither $(C;\pi \wedge x \neq t)$ nor $(C;\pi)\{x \mapsto t\}$, then trivially $(C_0,\pi_0)=(C_1,\pi_1)$.
	Otherwise, $(C_1,\pi_1)= (C;\pi \wedge x \neq t)$ or $(C_1,\pi_1)= (C;\pi)\{x \mapsto t\}$.
	Then $(L_1,{q_1})= (L,q)$ or $(L_1,{q_1})=(L,q)\{x \mapsto t\}$.
	In either case,(i) $L'\delta\vert_{q'}$ is an instance of $\skt(L\vert_q,\pi_0)$, 
	(ii) $q=q'$ if $L$ and $L'$ have the same predicate, and 
	(iii) if $L'\vert_{q'}=x$ and there exists an ancestor variable $y$ of $x$ in $(C_0,\pi_0)$, then $L\vert_{q}=y$.
	
	Let $N_0 \mo \Proj{P}(N)=N_1$.
	If $P$ is not the predicate of $L$, then trivially $(L,q)=(L_1,{q_1})$.
	If $P$ is the predicate of $L$, then $(L,q)=(P(t_1,\ldots,t_n),q)$ and $(L_1,{q_1})=(T(f_p(t_1,\ldots,t_n)),1.q)$.
	Thus, (i) $L'\delta\vert_{q'}$ is an instance of $\skt(L\vert_{q},\pi_0) =$ $\skt(T(f_p(t_1,\ldots,t_n)\vert_{1.q},\pi_0)$.
	(ii)  The predicate of $L'$ is not $P$ by definition.
	(iii) Let $L'\vert_{q'}=x$ and $y$ be the ancestor variable of $x$ in $(C_0,\pi_0)$.
	Then,  $y$ is also the ancestor variable of $x$ in $(C_1,\pi_1)$ and $L_1\vert_{q_1}=y$.
	Therefore, $L\vert_q=P(t_1,\ldots,t_n)\vert_{q}=T(f_p(t_1,\ldots,t_n)\vert_{1.q}=L_1\vert_{q_1}=y$.
	
	Let $N_0= N \cup \{ (C;\pi)\} \li N \cup \{ (C_a;\pi_a)\}=N_1  $ where an occurrence of a variable $x$ is replaced by a fresh $x'$.
	If $(C_1,\pi_1) \neq (C_a;\pi_a)$, then trivially $(C_0,\pi_0)=(C_1,\pi_1)$.
	Otherwise, $(C_1,\pi_1) = (C_a;\pi_a)$, $(C_0,\pi_0)=(C,\pi)$.
	By definition, $(L,q)=(L_1\{x' \mapsto x\},q_1)$ and $\pi_0=\pi_1\{x' \mapsto x\}$.
	Thus, $\skt(L\vert_{q},\pi_0)= \skt(L_1\vert_{q_1},\pi_1)$.  
	Therefore, $L'\delta\vert_{q'}$ is an instance of $\skt(L\vert_{q},\pi_0)$.
	Since $L$ and $L_1$ have the same predicate and $q=q_1$, $q=q'$ if $L$ and $L'$ have the same predicate.
	Let $L'\vert_{q'}=z$ and $y$ be the ancestor variable of $z$ in $(C_1,\pi_1)$.
	If $y \neq x'$, then $y$ is the ancestor variable of $z$ in $(C_0,\pi_0)$ and $L\vert_{q}=L_1\{x' \mapsto x\}\vert_{q_1}=y_1$.
	Otherwise, $x$ is the ancestor variable of $z$ in $(C_0,\pi_0)$ and $L\vert_{q}=L_1\{x' \mapsto x\}\vert_{q_1}=x$.
	
	Let $N_0= N \cup \{ (C;\pi)\} \sh N \cup \{ (C_l;\pi_l),(C_r;\pi_r)\}=N_1$ where a term $s$ is extracted from a positive literal $Q(s'[s]_p)$ via introduction of fresh predicate $S$ and variable $x$.
	If $(C_1,\pi_1)$ is neither $(C_l;\pi_l)$ nor $(C_r;\pi_r)$, then trivially $(C_0,\pi_0)=(C_1,\pi_1)$.
	
	If $(C_1,\pi_1)= (C_l;\pi_l)$ and $L_1 = S(x)$,
	then $(C_0,\pi_0)=(C;\pi)$, $q_1=1$, $(L',q')=(S(x),1)$ and $(Q(s'[s]_p),1.p)$ is the parent literal position of $(S(x),1)$.
	Let $L'\delta=S(t)$.
	Because $N^\bot_k$ is complete and ground, there is a clause $C'_k\delta'\in N^\bot_k$ that contains the positive literal $S(t)$.
	The ancestor of $(C'_k,\pi'_k)\in N_k$ in $N_1$ is $(C_r;\pi_r)$  because it is the only clause in $N_1$ with a positive $S$-literal. 
	Then, by the inductive hypothesis, $(S(s),1)$ in $(C_r;\pi_r)$ is the ancestor literal position of $(S(x),1)$ in $(C'_k,\pi'_k)$.
	Thus, $t$ is an instance of $\skt(S(s)\vert_1,\pi_r)=\skt(s,\pi_r)$.
	Therefore, $t=L'\delta\vert_{q'}$ is an instance of  $\skt(Q(s'[s]_p)\vert_{1.p},\pi)=\skt(s,\pi_r)$.
	Further, $Q$ and $S$ are not the same predicate because $S$ is fresh.
	Since $x$ has no parent variable, $L'\vert_{q'}=x$ has no ancestor variable in  $(C_0,\pi_0)$.
	
	If $(C_1,\pi_1)= (C_l;\pi_l)$ and $L_1 = Q(s'[p/x])$,
	then $(C_0,\pi_0)=(C;\pi)$ and $(Q(s'[s]_p),q_1)$ in $(C;\pi)$ is the parent literal position of $(L_1,q_1)$ in $(C_1,\pi_1)$ and ancestor literal position of $(L',q')$ in $(C_k,\pi_k)$.
	If  $q_1$ is not a position at or above $p$, the subterm at $p$ is irrelevant and thus  $\skt(Q(s'[s]_p)\vert_{q_1},\pi)=\skt(Q(s'[p/x])\vert_{q_1},\pi_l)$.
	Otherwise, let $r$ be a position such that $q_1r=1.p$.
	Since $\vert p \vert=2$, no following shallow transformation step extracts a subterm of $s'[p/x]$ containing $x$.
	Thus by definition of $\apr$, $L'=Q(t'[x]_p)$ and $C_k$ also contains the negative literal $S(x)$. 
	Let $S(x)\delta=S(t)$.
	Analogously to the previous case,  $t$ is an instance of $\skt(s,\pi_r)$.
	Combined with $L'\delta\vert_{q'}$ being an instance of $\skt(L_1\vert_{q_1},\pi_1)=\skt(Q(s'[p/x])\vert_{q_1},\pi_l)$ and  $L'\delta\vert_{1.p}=t$,
	$L'\delta\vert_{q'}$ is an instance of  $\skt(Q(s'[s]_p)\vert_{q},\pi)$.
	Since $L$ and $L_1$ have the same predicate and $q=q_1$, $q=q'$ if $L$ and $L'$ have the same predicate.
	Let $L'\vert_{q'}=z$ and $y$ in $(C_1,\pi_1)$ be the ancestor variable of $z$ in $(C_k,\pi_k)$.
	Since $x$ has no parent, $y\neq x$ and $y$ in $(C_0,\pi_0)$ is the ancestor variable of $z$.
	Therefore, $Q(s'[s]_p)\vert_{q_1}=y$ because $Q(s'[p/x])\vert_{q_1}=y$.
	
	If $(C_1,\pi_1)= (C_r;\pi_r)$ and $L_1 = S(s)$, let $q_1=1.q'_1$.
	Then, $(C_0,\pi_0)=(C;\pi)$ and $(L,q)=(Q(s'[s]_p),1.pq'_1)$ in $(C_0,\pi_0)$ is the parent literal position of $(L_1,q_1)$ in $(C_1,\pi_1)$.
	Thus, $L'\delta\vert_{q'}$ is an instance of  $\skt((Q(s'[s]_p)\vert_{1.pq'_1},\pi)=\skt(s\vert_{q'_1},\pi)=\skt(L_1\vert_{q_1},\pi_r)$. 
	Because $S$ is fresh, $Q$ is not the predicate of $L'$.
	Let $L'\vert_{q'}=z$ and $y$ in $(C_1,\pi_1)$ be the ancestor variable of $z$ in $(C_k,\pi_k)$.
	Then, $y$ in $(C_0,\pi_0)$ is the ancestor variable of $z$ 
	and $Q(s'[s]_p)\vert_{q}=s\vert_{q'_1}=y$ because $s\vert_{q'_1}=L_1\vert_{q_1}=y$.
	
	Otherwise, $(L_1,q_1)$ in $(C_0,\pi_0)$ is the parent literal position of $(L_1,q_1)$ in $(C_1,\pi_1)$, by definition. 
	Then, $\skt(L_1,\pi)=\skt(L_1,\pi_l)$ or $\skt(L_1,\pi)=\skt(L_1,\pi_r)$, respectively.\qed
\end{proof}

Next, we define the notion of descendants and descendant relations to connect lift-conflicts in ground conflicting cores with their corresponding ancestor clauses. 
The goal, hereby, is that if a ground clause $D$ is not a descendant of a clause in $N$, then it can never appear in a conflicting core of an approximation of $N$.

\begin{definition}[Descendants]\label{def:refine:descendant}
Let  $N \apr^* N'$,  $[(C;\pi),N] \anc^* [(C';\pi'),N']$ and  $D$ be a ground instance of $(C';\pi')$.
Then, we call $D$ a \emph{descendant} of $(C;\pi)$ and define the $[(C;\pi),N] \anc^*[(C';\pi'),N']$-descendant relation $\des{}{D}{}$ that maps literals in $D$ to literal positions in $(C;\pi)$ using the following rule:
$$ L'\delta \des{}{D}{} (L,r) \text{ if }  L'\delta\in D \text{ and } [r,L,(C;\pi),N] \anc^* [\varepsilon,L',(C';\pi'),N'] $$
\end{definition}

For the descendant relations it is of importance to note that while there are potentially infinite ways that a lift-conflict $C_c$ can be a descendant of an original clause  $(C;\pi)$,
there are only finitely many distinct descendant relations over $C_c$ and $(C;\pi)$.
This means, if a refinement transformation can prevent one distinct descendant relation without generating new distinct descendant relations (Lemma~\ref{lem:refinement:descendants}),
a finite number of refinement steps can remove the lift-conflict $C_c$ from the descendants of $(C;\pi)$  (Lemma~\ref{lem:refinement:refine}).
Thereby, preventing any conflicting cores containing $C_c$ from being found again.

A clause $(C;\pi)$ can have two descendants that are the same except for the names of the $S$-predicates introduced by shallow transformations.
Because the used approximation $N \apr^* N'$ is arbitrary and therefore also the choice of fresh $S$-predicates, 
if $D$ is a descendant of  $(C;\pi)$, then any clause $D'$ equal to $D$ up to a renaming of $S$-predicates is also a descendant of  $(C;\pi)$.
On the other hand, the actual important information about an $S$-predicate is which term it extracts.
Two descendants of $(C;\pi)$ might be identical but their $S$-predicate extract different terms in $(C;\pi)$.
For example, $P(a)\imp S(f(a))$ is a descendant of $P(x),P(y)\imp Q(f(x),g(f(x)))$ but might  extract either occurrence of $f(x)$.  
These cases are distinguished by their respective descendant relations.
In the example, we have either  $S(f(a)) \des{}{D}{} (Q(f(x),g(f(x))),1)$ or $S(f(a)) \des{}{D}{} (Q(f(x),g(f(x))),2.1)$.

\begin{lemma}\label{lem:refinement:descendants}
Let $N_0=N \cup \{(C;\pi)\} \refine  N \cup \{(C;\pi \wedge x \neq t),(C;\pi)\{x \mapsto t\}\}=N_1$ be a refinement transformation and $D$ a ground clause.
If there is a $[(C;\pi \wedge x \neq t),N_1] \anc^*[(C';\pi'),N_2]$- or $[(C;\pi)\{x \mapsto t\},N_1] \anc^*[(C';\pi'),N_2]$-descendant relation $\des{}{D}{1}$,
then there is an equal  $[(C;\pi),N_0] \anc^*[(C';\pi'),N_2]$-descendant relation $\des{}{D}{0}$.
\end{lemma}

\begin{proof}
	Let $L_D $ be a literal of $D$ and $L' \des{}{D}{1} (L,r)$.
	If $D$ is a descendant of $(C;\pi \wedge x \neq t)$, then $[r,L,(C;\pi \wedge x \neq t),N_1] \anc^* [\varepsilon,L',(C';\pi'),N_2]$.
	Because $[r,L,(C;\pi),N_0] \anc [r,L,(C;\pi \wedge x \neq t),N_1]$, $L' \des{}{D}{0} (L,r)$.
	If $D$ is a descendant of $(C;\pi)\{x \mapsto t\}$, the proof is analogous.\qed
\end{proof}

\begin{lemma}[Refinement]\label{lem:refinement:refine}
Let $N \apr  N'$ and $N^\bot$ be a complete ground conflicting core of $N'$. 
If $C_c\in N^\bot$ is a lift-conflict, then  there exists a finite refinement $N \refine^*  N_{R}$
such that for any approximation $N_{R} \apr^* N'_{R}$ and ground conflicting core $N^\bot_{R}$ of $N'_{R}$,
$C_c$ is not a lift-conflict in  $N^\bot_{R}$ modulo duplicate literal elimination.
\end{lemma}

\begin{proof}
	Let $(C_a,\pi_a)$ be the conflict clause of $C_c$ and  $(C;\pi)\in N$ be the parent clause of $(C_a,\pi_a)$.
	$C_c$ is a descendant of $(C;\pi)$ with the corresponding $[(C;\pi),N] \anc [(C_a;\pi_a),N']$-descendant relation $\des{}{C_c}{0}$.
	We apply induction on the number of distinct $[(C;\pi),N] \anc^* [(C';\pi'),N'']$-descendant relations $\des{}{C_c}{}$ for arbitrary approximations $N \apr^*  N''$.
	
	Since only the shallow and linear transformations can produce lift-conflicts,
	the clause $(C;\pi)$ is replaced by either a linearized clause $(C';\pi')$ or 
	two shallow clauses $(C_l;\pi)$ and $(C_r;\pi)$.
	Then, the conflict clause $(C_a;\pi_a)$ of $C_c$ is either the linearized $(C';\pi')$ or the resolvent of $(C_l;\pi)$ and $(C_r;\pi)$.
	In either case, $C_c=C_a\delta$ for some solution $\delta$ of $\pi_a$.
	Furthermore, there exists a substitution $\tau=\{x'_1 \mapsto x_1,\ldots,x'_n \mapsto x_n\}$  such that $(C;\pi)$ and $(C_a;\pi_a)\tau$ are equal modulo duplicate literal elimination.
	That is, $\tau= \{x'\mapsto x\}$ for a linear transformation and $\tau=\rho^{-1}$ for shallow transformation (Definition~\ref{def:approx:resolvent}).
	
	Assume $C_c=C_a\tau\sigma$ for some grounding substitution $\sigma$, where $\tau\sigma$ is a solution of $\pi_a$.
	Thus, $\sigma$ is a solution of $\pi_a\tau$, which is equivalent to $\pi$.
	Then, $C_c$ is equal to $C\sigma$ modulo duplicate literal elimination an instance of $(C;\pi)$,
	which contradicts with $C_c$ being a lift-conflict.
	Hence, $C_c=C_a\delta$ is not an instance of $C_a\tau$ and thus, $x_i\delta \neq x'_i\delta$ for some $x_i$ in the domain of $\tau$.
	
	Because $x_i\delta $ and $ x'_i\delta$ are ground, there is a position $p$ where $x_i\delta\vert_p $ and $x'_i\delta\vert_p$ have different function symbols.
	We construct the straight term $t$ using the path from the root to $p$ on $x_i\delta$ with variables that are fresh in $(C,\pi)$.   
	Then, we can use $x_i$ and $t$ to segment $(C;\pi)$ into $(C;\pi \wedge x_i \neq t)$ and $(C;\pi)\{x_i \mapsto t\}$ for the refinement $N \refine  N_{R}$.
	Note, that $x_i\delta$ is a ground instance of $t$, while  $x'_i\delta$  is not.
	
	Let $(L'_1,r'_1)$ and $(L'_2,r'_2)$ in $(C_a,\pi_a)$ be literal positions of the variables $x_i$ and $x'_i$ in $C_a$,
	and $(L_1,r_1)$ and $(L_2,r_2)$ in $(C,\pi)$ be the parent literal positions  of $(L'_1,r'_1)$ and $(L'_2,r'_2)$, respectively. 
	Because $(C_a,\pi_a)\tau$ is equal to $(C;\pi)$ modulo duplicate literal elimination, $L_1\vert_{r_1} = L_2\vert_{r_2}=x_i$. 
	Let $N \refine N_1$ be the refinement where $(C;\pi)$ is segmented into $(C;\pi \wedge x_i \neq t)$ and $(C;\pi)\{x_i \mapsto t\}$. 
	
	By Lemma~\ref{lem:refinement:descendants}, all $[(C;\pi \wedge x_i \neq t),N_1] \anc^* [(C'_a;\pi'_a),N_2]$- or $[(C;\pi)\{x_i \mapsto t\},N_1] \anc^* [(C'_a;\pi'_a),N_2]$-descendant relations
	correspond to an equal $[(C;\pi),N] \anc [(C'_a;\pi'_a),N_2]$-descendant relation.
	Assume there is a $[(C;\pi \wedge x_i \neq t),N_1] \anc^* [(C'_a;\pi'_a),N_2]$-descendant relation $\des{}{C_c}{1}$ that is not distinct from  $\des{}{C_c}{0}$.
	Because  $L'_1\delta \des{}{C_c}{0} (L_1,r)$ for some literal position $(L_1,r)$ in $(C;\pi)$, which is the parent literal position of  $(L_1,r)$ in $(C;\pi \wedge x_i \neq t)$,
	$L'_1\delta \des{}{C_c}{1} (L_1,r)$.
	However, this contradicts Lemma~\ref{lem:refinement:ancestor_skel} because $x_i\delta $ is not an instance of $\skt(L_1\vert_{r_1},\pi \wedge x_i \neq t)=\skt(x_i,\pi \wedge x_i \neq t)$.
	The case that there is a  $[(C;\pi)\{x_i \mapsto t\},N_1] \anc^* [(C'_a;\pi'_a),N_2]$-descendant relation that is not distinct from  $\des{}{C_c}{0}$ is analogous using the argument that $x'_i\delta $ is not an instance of $\skt(L_2\{x_i \mapsto t\}\vert_{r_2},\pi)=\skt(t,\pi)$.
	Hence, there are strictly less distinct descendant relations over $C_c$ and $(C;\pi \wedge x \neq t)$ or $(C;\pi)\{x \mapsto t\}$
	than there are distinct descendant relations over $C_c$ and $(C,\pi)$.
	
	If there are no descendant relations, then $C_c$ can no longer appear as a lift conflict.  
	Otherwise, by the inductive hypothesis, there exists a finite refinement $N \refine  N_1 \refine^*  N_{R}$
	such that for any approximation $N_{R} \apr N'_{R}$ and ground conflicting core $N^\bot_{R}$ of $N'_{R}$,
	$C_c$ is not a lift-conflict in $N^\bot_{R}$ modulo duplicate literal elimination.\qed
\end{proof}

\begin{theorem}[Soundness and Completeness of FO-AR]\label{theo:refinement:scfoar}
Let $N$ be an unsatisfiable clause set and $N'$ its MSL(SDC) approximation: (i)~if $N$ is unsatisfiable then there exists a conflicting
core of $N'$ that can be lifted to a refutation in $N$, (ii)~if $N'$ is satisfiable, then $N$ is satisfiable too.
\end{theorem}
\begin{proof}(Idea) 
	By Lemma~\ref{lem:approx:sound} and Lemma~\ref{lem:refinement:refine}, where the latter can be used to show that a core of $N'$ that cannot
	be lifted also excludes the respective instance for unsatisfiability  of $N$.
	
	Let $(C_a,\pi_a)$ be the conflict clause of $C_c$ and  $(C;\pi)\in N$ be the parent clause of $(C_a,\pi_a)$.
	$C_c$ is a descendant of $(C;\pi)$ with the corresponding $[(C;\pi),N] \anc [(C_a;\pi_a),N']$-descendant relation $\des{}{C_c}{0}$.
	We apply induction on the number of distinct $[(C;\pi),N] \anc^* [(C';\pi'),N'']$-descendant relations $\des{}{C_c}{}$ for arbitrary approximations $N \apr^*  N''$.
	
	Since only the shallow and linear transformations can produce lift-conflicts,
	the clause $(C;\pi)$ is replaced by either a linearized clause $(C';\pi')$ or 
	two shallow clauses $(C_l;\pi)$ and $(C_r;\pi)$.
	Then, the conflict clause $(C_a;\pi_a)$ of $C_c$ is either the linearized $(C';\pi')$ or the resolvent of $(C_l;\pi)$ and $(C_r;\pi)$.
	In either case, $C_c=C_a\delta$ for some solution $\delta$ of $\pi_a$.
	Furthermore, there exists a substitution $\tau=\{x'_1 \mapsto x_1,\ldots,x'_n \mapsto x_n\}$  such that $(C;\pi)$ and $(C_a;\pi_a)\tau$ are equal modulo duplicate literal elimination.
	That is, $\tau= \{x'\mapsto x\}$ for a linear transformation and $\tau=\rho^{-1}$ for shallow transformation (Definition~\ref{def:approx:resolvent}).
	
	Assume $C_c=C_a\tau\sigma$ for some grounding substitution $\sigma$, where $\tau\sigma$ is a solution of $\pi_a$.
	Thus, $\sigma$ is a solution of $\pi_a\tau$, which is equivalent to $\pi$.
	Then, $C_c$ is equal to $C\sigma$ modulo duplicate literal elimination an instance of $(C;\pi)$,
	which contradicts with $C_c$ being a lift-conflict.
	Hence, $C_c=C_a\delta$ is not an instance of $C_a\tau$ and thus, $x_i\delta \neq x'_i\delta$ for some $x_i$ in the domain of $\tau$.
	
	Because $x_i\delta $ and $ x'_i\delta$ are ground, there is a position $p$ where $x_i\delta\vert_p $ and $x'_i\delta\vert_p$ have different function symbols.
	We construct the straight term $t$ using the path from the root to $p$ on $x_i\delta$ with variables that are fresh in $(C,\pi)$.   
	Then, we can use $x_i$ and $t$ to segment $(C;\pi)$ into $(C;\pi \wedge x_i \neq t)$ and $(C;\pi)\{x_i \mapsto t\}$ for the refinement $N \refine  N_{R}$.
	Note, that $x_i\delta$ is a ground instance of $t$, while  $x'_i\delta$  is not.
	
	Let $(L'_1,r'_1)$ and $(L'_2,r'_2)$ in $(C_a,\pi_a)$ be literal positions of the variables $x_i$ and $x'_i$ in $C_a$,
	and $(L_1,r_1)$ and $(L_2,r_2)$ in $(C,\pi)$ be the parent literal positions  of $(L'_1,r'_1)$ and $(L'_2,r'_2)$, respectively. 
	Because $(C_a,\pi_a)\tau$ is equal to $(C;\pi)$ modulo duplicate literal elimination, $L_1\vert_{r_1} = L_2\vert_{r_2}=x_i$. 
	Let $N \refine N_1$ be the refinement where $(C;\pi)$ is segmented into $(C;\pi \wedge x_i \neq t)$ and $(C;\pi)\{x_i \mapsto t\}$. 
	
	By Lemma~\ref{lem:refinement:descendants}, all $[(C;\pi \wedge x_i \neq t),N_1] \anc^* [(C'_a;\pi'_a),N_2]$- or $[(C;\pi)\{x_i \mapsto t\},N_1] \anc^* [(C'_a;\pi'_a),N_2]$-descendant relations
	correspond to an equal $[(C;\pi),N] \anc [(C'_a;\pi'_a),N_2]$-descendant relation.
	Assume there is a $[(C;\pi \wedge x_i \neq t),N_1] \anc^* [(C'_a;\pi'_a),N_2]$-descendant relation $\des{}{C_c}{1}$ that is not distinct from  $\des{}{C_c}{0}$.
	Because  $L'_1\delta \des{}{C_c}{0} (L_1,r)$ for some literal position $(L_1,r)$ in $(C;\pi)$, which is the parent literal position of  $(L_1,r)$ in $(C;\pi \wedge x_i \neq t)$,
	$L'_1\delta \des{}{C_c}{1} (L_1,r)$.
	However, this contradicts Lemma~\ref{lem:refinement:ancestor_skel} because $x_i\delta $ is not an instance of $\skt(L_1\vert_{r_1},\pi \wedge x_i \neq t)=\skt(x_i,\pi \wedge x_i \neq t)$.
	The case that there is a  $[(C;\pi)\{x_i \mapsto t\},N_1] \anc^* [(C'_a;\pi'_a),N_2]$-descendant relation that is not distinct from  $\des{}{C_c}{0}$ is analogous using the argument that $x'_i\delta $ is not an instance of $\skt(L_2\{x_i \mapsto t\}\vert_{r_2},\pi)=\skt(t,\pi)$.
	Hence, there are strictly less distinct descendant relations over $C_c$ and $(C;\pi \wedge x \neq t)$ or $(C;\pi)\{x \mapsto t\}$
	than there are distinct descendant relations over $C_c$ and $(C,\pi)$.
	
	If there are no descendant relations, then $C_c$ can no longer appear as a lift conflict.  
	Otherwise, by the inductive hypothesis, there exists a finite refinement $N \refine  N_1 \refine^*  N_{R}$
	such that for any approximation $N_{R} \apr N'_{R}$ and ground conflicting core $N^\bot_{R}$ of $N'_{R}$,
	$C_c$ is not a lift-conflict in $N^\bot_{R}$ modulo duplicate literal elimination.\qed
\end{proof}

Actually,  Lemma~\ref{lem:refinement:refine} can be used to define a fair strategy on refutations in $N'$ in order
to receive also a dynamically complete FO-AR calculus, following the ideas presented in \cite{Teucke2015}.

In Lemma~\ref{lem:refinement:refine}, we segment the conflict clause's immediate parent clause.
If the lifting later successfully passes this point, the refinement is lost and will be possibly repeated.
Instead, we can refine any ancestor of the conflict clause as long as it contains the ancestor of the variable used in the refinement.
By Lemma~\ref{lem:refinement:ancestor_skel}-(iii), such an ancestor will contain the ancestor variable at the same positions. 
If we refine the ancestor in the original clause set, the refinement is permanent because lifting the refinement steps always succeeds.
Only variables introduced by shallow transformation cannot be traced to the original clause set.
However, these shallow variables are already linear and the partitioning in the shallow transformation can be chosen such that they are not shared variables.
Assume a shallow, shared variable $y$, that is used to extract the term $t$, in the shallow transformation 
of $\Gamma\imp E[s]_p,\Delta$ into $S(x),\Gamma_l\imp E[p/x],\Delta_l$ and $\Gamma_r\imp S(s),\Delta_r$.
Since $\Delta_l$ $\dot{\cup}$ $\Delta_r=\Delta$ is a partitioning, $y$ can only appear in either   $E[p/x],\Delta_l$ or $S(s),\Delta_r$.
If $y \in \vars(E[p/x],\Delta_l)$ we instantiate $\Gamma_r$ with $\{y \mapsto t\}$ and  $\Gamma_l$, otherwise.
Now, $y$ is no longer a shared variable.\\

The refinement Lemmas only guarantee a refinement for a given ground conflicting core.
In practice, however, conflicting cores contain free variables. 
We can always generate a ground conflicting core by instantiating the free variables with ground terms.
However, if we only exclude a single ground case via refinement, next time the new conflicting core will likely have overlaps with the previous one.
Instead, we can often remove all ground instances of a given conflict clause at once.   

The simplest case is when unifying the conflict clause with the original clause fails because their instantiations differ at some equivalent positions.
For example, consider $N= \{ P(x,x); P(f(x,a), f(y,b)) \imp\}$. 
$N$ is satisfiable but the linear transformation is unsatisfiable with conflict clause $P(f(x,a), f(y,b))$ 
which is not unifiable with $P(x,x)$, because the two terms $f(x,a)$ and $f(y,b)$ have different constants at the second argument.
A refinement of $ P(x,x)$ is \newline
\centerline{$\begin{array}{r@{\,;\,}l}
 ( P(x,x) & x \neq f(v,a)) \\
(P(f(x,a),f(x,a)) & \top) \\
\end{array}$}
 $ P(f(x,a), f(y,b))$ shares no ground instances with the approximations of the refined clauses. 

Next, assume that again unification fails due to structural difference, but this time the differences lie at different positions.
For example, consider $N= {\{ P(x,x); P(f(a,b), f(x,x)) \imp\}}$. 
$N$ is satisfiable but the linear transformation of $N$ is unsatisfiable with conflict clause $  P(f(a,b), f(x,x))$
 which is not unifiable with $ P(x,x)$ because in $f(a,b)$ the first an second argument are different but the same in $f(x,x)$.
A refinement of $ P(x,x)$ is \newline
\centerline{$\begin{array}{r@{\,;\,}l}
	( P(x,x) & x \neq f(a,v)) \\
	(P(f(a,x), f(a,x))) & x \neq a) \\
	( P(f(a,a), f(a,a))) & \top)  \\
	\end{array}$}
 $  P(f(a,b), f(x,x))$ shares no ground instances with the approximations of the refined clauses. 

It is also possible that the conflict clause and original clause are unifiable by themselves, but the resulting constraint has no solutions. 
For example, consider $N= {\{  P(x,x); (P(x, y) \imp; x \neq a \wedge x \neq b \wedge y \neq c \wedge y \neq d  )\}}$ with signature $\Sigma=\{a,b,c,d\}$. 
$N$ is satisfiable but the linear transformation of $N$ is unsatisfiable with conflict clause $ (\imp P(x,y); x \neq a \wedge x \neq b \wedge y \neq c \wedge y \neq d)$.
While $P(x,x)$ and $P(x,y)$ are unifiable, the resulting constraint $ x \neq a \wedge x \neq b \wedge x \neq c  \wedge x \neq d$ has no solutions.
A refinement of $P(x,x)$ is \newline
\centerline{$\begin{array}{r@{\,;\,}l}
	(P(x,x) & x \neq a \wedge x \neq b) \\
	( P(a,a) & \top) \\
	(P(b,b) & \top) \\
	\end{array}$}
 $ ( P(x,y); x \neq a \wedge x \neq b \wedge y \neq c \wedge y \neq d)$ shares no ground instances with the approximations of the refined clauses.\\

Lastly, we should mention that there are cases where the refinement process does not terminate.
For example, consider the clause set $N= {\{  P(x,x) ; P(y,g(y)) \imp \}}$. 
$N$ is satisfiable but the linear transformation of $N$ is unsatisfiable with conflict clause $ P(y,g(y))$, 
which is not unifiable with $ P(x,x)$.
A refinement of $ P(x,x)$ based on the ground instance $P(a,g(a))$ is \newline
\centerline{$\begin{array}{r@{\,;\,}l}
 ( P(x,x) & x \neq g(v)) \\
( P(g(x),g(x)) & \top) \\
\end{array}$}
While $ P(y,g(y))$ is not an instance of the refined approximation,
it shares ground instances with $ P(g(x),g(x'))$. 
The new conflict clause is $ P(g(y),g(g(y)))$ and 
the refinement will continue to enumerate all $ P(g^i(x),g^i(x))$ instances of $ P(x,x)$ without ever reaching a satisfiable approximation. 
Satisfiability of first-order clause sets is undecidable, so termination cannot be expected by any calculus, in general.

\section{Experiments}\label{sec:experiments}

In the following we discuss several first-order clause classes for which FO-AR implemented in SPASS-AR immediately
decides satisfiability but
superposition and instantiation-based methods fail. 
We argue both according to the respective calculi and state-of-the-art implementations,
in particular SPASS~3.9~\cite{DBLP:conf/cade/WeidenbachDFKSW09}, Vampire~4.1~\cite{KovacsVoronkov13,DBLP:conf/cav/Voronkov14}, 
for ordered-resolution/superposition,
iProver~2.5~\cite{DBLP:conf/cade/Korovin08} an  implementation of Inst-Gen \cite{Korovin13ganzinger}, and
Darwin~v1.4.5 \cite{DBLP:journals/ijait/BaumgartnerFT06} an implementation of the  model evolution calculus \cite{DBLP:conf/cade/BaumgartnerT03}.
All experiments were run on
a 64-Bit Linux computer (Xeon(R) E5-2680, 2.70GHz, 256GB main memory).
For Vampire and Darwin we chose the CASC-sat and CASC settings, respectively. 
For iProver we set the schedule to ``sat'' and SPASS, SPASS-AR were used in default mode.
Please note that Vampire and iProver are portfolio solvers including implementations of several different calculi including
superposition (ordered resolution), instance generation, and finite model finding.
SPASS, SPASS-AR, and Darwin only implement superposition, FO-AR, and model evolution, respectively.

For the first example\newline
\centerline{$P(x,y) \imp P(x,z) , P(z,y);\quad P(a,a)$}
and second example,\newline
\centerline{$Q(x,x);\quad
 Q(v,w) , P(x,y)  \imp P(x,v) , P(w,y);\quad
 P(a,a)$}
the superposition calculus produces independently of the selection strategy and ordering
an infinite number of clauses of form\newline
\centerline{$\begin{array}{r@{\,\imp\,}l}
&P(a,z_1),\; P(z_1,z_2),\;\ldots,\;P(z_n,a).\\
\end{array}$}

Using linear approximation, however, FO-AR replaces  $P(x,y) \imp P(x,z) , P(z,y)$ and $\imp Q(x,x)$ with
$P(x,y) \imp P(x,z) , P(z',y)$ and $\imp Q(x,x')$, respectively.
Consequently, ordered resolution derives $\imp P(a,z_1) , P(z_2,a)$ which subsumes any further inferences $\imp P(a,z_1)  , P(z_2,z_3) , P(z_4,a)$.
Hence, saturation of the approximation terminates immediately.
Both examples belong to the Bernays-Sch\"onfinkel fragment, so
model evolution (Darwin) and Inst-Gen (iProver) can decide them as well. 
Note that the concrete behavior of superposition is not limited to the above examples 
but potentially occurs whenever there are variable chains in clauses. 

On the third problem\newline
\centerline{$P(x,y) \imp P(g(x),z);\quad
P(a,a)$}
superposition derives all clauses of the form $\imp P(g(\ldots g(a)\ldots),z)$.
With a shallow approximation of $ P(x,y) \imp P(g(x),z)$ into $ S(v) \imp P(v,z)$ and $ P(x,y) \imp S(g(x))$,
FO-AR (SPASS-AR) terminates after deriving $ \imp S(g(a))$ and $S(x) \imp S(g(x))$.
Again, model evolution (Darwin)  and Inst-Gen (iProver) can also solve this example.

The next example\newline
\centerline{$P(a);\quad P(f(a))\imp;\quad
  P(f(f(x)))  \imp P(x);\quad
  P(x) \imp P(f(f(x)))$}
is already saturated under superposition. 
For FO-AR,
the clause $P(x) \imp P(f(f(x)))$ is replaced by $S(x) \imp P(f(x))$ and $P(x) \imp S(f(x))$.
Then ordered resolution terminates after inferring $S(a) \imp$ and $S(f(x)) \imp P(x)$.

The Inst-Gen and model evolution calculi, however, fail. 
In either, a satisfying model is represented by a finite set of literals, i.e,
a model of the propositional approximation for Inst-Gen 
and the trail of literals in case of model evolution.
Therefore, there necessarily exists a literal $P(f^n(x))$ or $\neg P(f^n(x))$ with a maximal $n$ in these models.  
This contradicts the actual model where either $P(f^n(a))$ or $P(f^n(f(a)))$ is true.
However, iProver can solve this problem using its built-in ordered resolution solver whereas
Darwin does not terminate on this problem.

Lastly consider an example of the form\newline
\centerline{$f(x) \approx x  \imp;\,\;
f(f(x)) \approx x  \imp;\,\, \ldots;
f^n(x) \approx x  \imp$}
which is trivially satisfiable, e.g., saturated by superposition, but any model has at least $n+1$ domain elements.
Therefore, adding these clauses to any satisfiable clause set containing $f$ forces calculi that explicitly
consider finite models to consider at least $n+1$ elements. The performance of final model finders~\cite{SlaneySurendonk96} typically degrades
in the number of different domain elements to be considered.

Combining each of these examples into one problem is then 
solvable by neither superposition, Inst-Gen, or model evolution
and not practically solvable with increasing $n$ via testing finite models.
For example, we tested\newline
\centerline{$\begin{array}{c}
  P(x,y) \imp P(x,z) , P(z,y); \quad P(a,a); \quad P(f(a),y) \imp; \\
  P(f(f(x)),y) \imp P(x,y);\quad P(x,y) \imp P(f(f(x)),y);  \\
  f(x) \approx x \imp;, \ldots, f^{n}(x) \approx x \imp;\\
  \end{array}$}
for $n=20$ against SPASS, 
Vampire, 
iProver, 
and Darwin for more than one hour each without success.
Only SPASS-AR solved it in less than one second. 

For iProver we added an artificial positive equation $b\approx c$. For otherwise,
iProver throws away all disequations while preprocessing. This is a satisfiability
preserving operation, however, the afterwards found (finite) models are not models of the
above clause set due to the collapsing of ground terms.

\section{Conclusion} \label{sec:conclusion}

The previous section showed FO-AR is superior to superposition,
instantiation-based methods on certain classes of clause sets. Of course,
there are also classes of clause sets where superposition and instantiation-based methods
are superior to FO-AR, e.g., for unsatisfiable clause sets where the structure
of the clause set forces FO-AR to enumerate failing ground instances
due to the approximation in a bottom-up way.

Our prototypical implementation SPASS-AR cannot compete with systems such as iProver
or Vampire on the respective CASC categories of the TPTP~\cite{Sut09}. This 
is already due to the fact that they
are all meanwhile portfolio solvers. For example, iProver contains an implementation
of ordered resolution and Vampire an implementation of Inst-Gen. Our results, Section~\ref{sec:experiments}, however,
show that these systems may benefit from FO-AR by adding it to their portfolio.

The DEXPTIME-completeness result for MSLH strongly suggest that both the MSLH and also our
MSL(SDC) fragment have the finite model property. However, we are not aware of
any proof. If MSL(DSC) has the finite model property, the finite model finding approaches are complete
on MSL(SDC). The models generated by FO-AR and superposition
are typically infinite.
It remains an open problem, even for fragments enjoying the finite model property,
e.g., the first-order monadic fragment, to design a calculus that combines explicit finite
model finding with a structural representation of infinite models.
For classes that have no finite models this problem seems to become even more 
difficult. To the best of our knowledge, SPASS is currently the only prover
that can show satisfiability of the clauses $R(x,x)\imp$; $R(x,y), R(y,z)\imp R(x,z)$;
$R(x,g(x))$ due to an implementation of chaining~\cite{BachmairGanzinger98,SudaEtAl10}. 
Apart from the superposition calculus, it is unknown to us how the
specific inferences for transitivity can be combined with any of the other discussed calculi,
including the abstraction refinement calculus introduced in this paper.

Finally, there are not many results on calculi that operate with respect to
models containing positive equations. Even for fragments that are decidable
with equality, such as the Bernays-Schoenfinkel-Ramsey fragment or the monadic
fragment with equality, there seem currently no convincing suggestions compared
to the great amount of techniques for these fragments without equality. Adding
positive equations to MSL(SDC) while keeping decidability is, to the best of
our current knowledge, only possible for at most linear, shallow equations
$f(x_1,\ldots,x_n) \approx h(y_1,\ldots,y_n)$~\cite{JacquemardMeyerEtAl98}. However, approximation into
such equations from an equational theory with nested term occurrences typically
results in an almost trivial equational theory. So this does not seem to be 
a very promising research direction.


\bibliographystyle{plain}
\bibliography{abstractions}

\end{document}